%% file: papier.tex
\documentclass[11pt]{article}

\usepackage{amssymb,amsmath,amsthm}
\usepackage{enumerate}
\usepackage{epsfig}
\usepackage{version}
\usepackage{fullpage}
\usepackage{url}

\newtheorem{theorem}{Theorem}[section]

\newtheorem{observation}[theorem]{Observation}

\newtheorem{definition}[theorem]{Definition}

\newtheorem{corollary}[theorem]{Corollary}
\newtheorem{claim}[theorem]{Claim}
\newtheorem{lemma}[theorem]{Lemma}
\newtheorem{polyrule}{Rule}[section]

\newcommand{\PIC}{\textsc{Proper Interval Completion}}
\newcommand{\BCC}{\textsc{Bi-clique Chain Completion}}
\newcommand{\BCD}{\textsc{Bipartite Chain Deletion}}

\title{Polynomial kernels for \textsc{Proper Interval Completion} and related problems
\thanks{Research supported by the AGAPE project  (ANR-09-BLAN-0159).}}

\author{St\'ephane Bessy \hspace{0.4cm} Anthony Perez\\
bessy@lirmm.fr \hspace{0.4cm} perez@lirmm.fr \\
  LIRMM -- Universit\'e Montpellier II - FRANCE}

\begin{document}

\maketitle

\begin{abstract}
Given a graph $G = (V,E)$ and a positive integer $k$, the \PIC{}
problem asks whether there exists a set $F$ of at most $k$ pairs of
$(V \times V) \setminus E$ such that the graph $H = (V, E \cup F)$ is
a proper interval graph. The \PIC{} problem finds applications in
molecular biology and genomic research~\cite{KST94,Sha02}. First
announced by Kaplan, Tarjan and Shamir in FOCS '94, this problem is
known to be FPT~\cite{KST94}, but no polynomial kernel was known to
exist. We settle this question by proving that \PIC{} admits a kernel
with at most $O(k^5)$ vertices. Moreover, we prove that a related
problem, the so-called \textsc{Bipartite Chain Deletion} problem,
admits a kernel with at most $O(k^2)$ vertices, completing a previous
result of Guo~\cite{Guo07}.
\end{abstract}

\section*{Introduction}
	
% GRAPH MODIFICATION PROBLEMS 

The aim of a graph modification problem is to transform a given graph
in order to get a certain property $\Pi$ satisfied. Several types of
transformations can be considered: for instance, in \emph{vertex
  deletion} problems, we are only allowed to delete vertices from the
input graph, while in \emph{edge modification problems} the only
allowed operation is to modify the edge set of the input graph. The
optimization version of such problems consists in finding a
\emph{minimum} set of edges (or vertices) whose modification makes the
graph satisfy the given property $\Pi$.  Graph modification problems
cover a broad range of NP-Complete problems and have been extensively
studied in the literature~\cite{Man08,SST04,Sha02}.  Well-known examples
include the \textsc{Vertex Cover}~\cite{DFRS04}, \textsc{Feedback
  Vertex Set}~\cite{Tho10}, or \textsc{Cluster Editing}~\cite{CM10}
problems.  These problems find applications in various domains, such
as computational biology~\cite{KST94,Sha02}, image
processing~\cite{SST04} or relational databases~\cite{TY84}.   
	
% PARAMETERIZED COMPLEXITY 

Due to these applications, one may be interested in computing an
\emph{exact} solution for such problems. \emph{Parameterized
  complexity} provides a useful theoretical framework to that
aim~\cite{DF99,Nie06}. A problem \emph{parameterized} by some integer
$k$ is said to be \emph{fixed-parameter tractable} (FPT for short)
whenever it can be solved in time $f(k) \cdot n^c$ for any constant $c
> 0$. A natural parameterization for graph modification problems
thereby consists in the number of allowed transformations.  As one of
the most powerful technique to design fixed-parameter algorithms,
\emph{kernelization algorithms} have been extensively studied in the
last decade (see~\cite{Bod09} for a survey). A \emph{kernelization
  algorithm} is a polynomial-time algorithm (called \emph{reduction
  rules}) that given an instance $(I,k)$ of a parameterized problem
$P$ computes an instance $(I', k')$ of $P$ such that $(i)$ $(I,k)$ is
a \textsc{Yes}-instance if and only if $(I',k')$ is a
\textsc{Yes}-instance and $(ii)$ $|I'| \le h(k)$ for some computable
function $h()$ and $k' \le k$.  The instance $(I',k')$ is called the
\emph{kernel} of $P$. We say that $(I', k')$ is a polynomial kernel if
the function $h()$ is a polynomial. It is well-known that a
parameterized problem is FPT if and only if it has a kernelization
algorithm~\cite{Nie06}. But this equivalence only yields kernels of
super-polynomial size.  To design efficient fixed-parameter
algorithms, a kernel of small size - polynomial (or even linear) in
$k$ - is highly desirable~\cite{NR00}. However, recent results give
evidence that not every parameterized problem admits a polynomial
kernel, unless $NP \subseteq coNP/poly$~\cite{BDFH09}. On the positive
side, notable kernelization results include a less-than-$2k$ kernel for
\textsc{Vertex Cover}~\cite{DFRS04}, a $4k^2$ kernel for
\textsc{Feedback Vertex Set}~\cite{Tho10} and a $2k$ kernel for
\textsc{Cluster Editing}~\cite{CM10}.
	
% GRAPH MODIFICATION PROBLEMS

We follow this line of research with respect to graph modification
problems.  It has been shown that a graph modification problem is FPT
whenever $\Pi$ is hereditary and can be characterized by a finite set
of forbidden induced subgraphs~\cite{Cai96}. However, recent results
proved that several graph modification problems do not admit a
polynomial kernel even for such properties $\Pi$~\cite{GPP10,KW09}.
In this paper, we are in particular interested in \emph{completion}
problems, where the only allowed operation is to add edges to the
input graph. We consider the property $\Pi$ as being the class of
\emph{proper interval graphs}.  This class is a well-studied class of
graphs, and several characterizations are known to
exist~\cite{LO93,Weg67}.  In particular, there exists an
\emph{infinite} set of forbidden induced subgraphs that characterizes
proper interval graphs~\cite{Weg67} (see Figure~\ref{fig:forbidden}).
More formally, we consider the following problem: \\
\\ 
\PIC{}:
\\ 
\textbf{Input}: A graph $G = (V,E)$ and a positive integer
$k$. \\ 
\textbf{Parameter}: $k$. \\ 
\textbf{Output}: A set $F$ of at
most $k$ pairs of $(V \times V) \setminus E$ such that the graph $H =
(V, E \cup F)$ is a proper interval graph.\\

% motivations 
Interval completion problems find applications in
molecular biology and genomic research~\cite{HSS01,KST94}, and in
particular in \emph{physical mapping} of DNA. In this case, one is
given a set of long contiguous intervals (called \emph{clones})
together with experimental information on their pairwise overlaps, and
the goal is to reconstruct the relative position of the clones along
the target DNA molecule. We focus here on the particular case where
all intervals have equal length, which is a biologically important
case (e.g. for cosmid clones~\cite{HSS01}).  In the presence of (a
small number of) unidentified overlaps, the problem becomes equivalent
to the \PIC{} problem. 
 It is known to be NP-Complete for a long
time~\cite{GKS94}, but fixed-parameter tractable due to a result of
Kaplan, Tarjan and Shamir in FOCS '94~\cite{KST94, KST99}.
\footnote{Notice also that the \emph{vertex deletion} of the problem
  is fixed-parameter tractable~\cite{Vil10}.}
The fixed-parameter tractability of the \PIC{} can also be seen 
as a corollary of a 
characterization of Wegner~\cite{Weg67} combined with Cai's
result~\cite{Cai96}. Nevertheless, it was not known whether this
problem admit a polynomial kernel or not.

% OUR RESULTS	

\paragraph{Our results} We prove that the \PIC{} problem admits a kernel 
with at most $O(k^5)$ vertices. To that aim, we identify \emph{nice}
parts of the graph that induce proper interval graphs and can hence be
safely reduced. Moreover, we apply our techniques to the so-called
\textsc{Bipartite Chain Deletion} problem, closely related to the
\PIC{} problem where one is given a graph $G = (V,E)$ and seeks a set
of at most $k$ edges whose deletion from $E$ result in a bipartite
chain graph (a graph that can be partitioned into two independent sets
connected by a join). We obtain a kernel with $O(k^2)$ vertices for
this problem. This result completes a previous result of
Guo~\cite{Guo07} who proved that the \textsc{Bipartite Chain Deletion
  With Fixed Bipartition} problem admits a kernel with $O(k^2)$
vertices.
	
%OUTLINES 
\paragraph{Outline} We begin with some definitions and notations regarding 
proper interval graphs. Next, we give the reduction rules the
application of which leads to a kernelization algorithm for the \PIC{}
problem. These reduction rules allow us to obtain a kernel with at most 
$O(k^5)$ vertices. Finally, we prove that our techniques can be applied to
\BCD{} to obtain a quadratic-vertex kernel, completing a previous 
result of Guo~\cite{Guo07}.

% ======================================================================= %

\section{Preliminaries}

\subsection{Proper interval graphs}
	
We consider simple, loopless, undirected graphs $G = (V,E)$ where
$V(G)$ denotes the vertex set of $G$ and $E(G)$ its edge
set\footnote{In all our notations, we forget the mention to the graph
  $G$ whenever the context is clear.}.  Given a vertex $v \in V$, we
use $N_G(v)$ to denote the \emph{open neighborhood} of $v$ and $N_G[v]
= N_G(v) \cup \{v\}$ for its \emph{closed neighborhood}. Two vertices
$u$ and $v$ are \emph{true twins} if $N[u] = N[v]$.  If $u$ and $v$
are not true twins but $uv \in E$, we say that a vertex of
$N[u]\bigtriangleup N[v]$ \emph{distinguishes} $u$ and $v$.  Given a
subset of vertices $S \subseteq V$, $N_S(v)$ denotes the set $N_G(v)
\cap S$ and $N_G(S)$ denotes the set $\{N_G(s) \setminus S : s \in
S\}$. Moreover, $G[S]$ denotes the subgraph \emph{induced} by $S$,
i.e. $G[S] = (S, E_S)$ where $E_S = \{uv \in E\ :\ u,v \in S\}$.  A
\emph{join} in a graph $G=(V,E)$ is a bipartition $(X,Y)$ of $G$ and an
order $x_1,\dots ,x_{|X|}$ on $X$ such that for all $i=1,\dots
,|X|-1$, $N_Y(x_i)\subseteq N_Y(x_{i+1})$. The edges between $X$ and
$Y$ are called the \emph{edges of the join}, and a subset $F \subseteq
E$ is said to \emph{form a join} if $F$ corresponds to the edges of a
join of $G$.  Finally, a graph is an \emph{interval graph} if it
admits a representation on the real line such that: $(i)$ the vertices
of $G$ are in bijection with intervals of the real line and $(ii)$ $uv
\in E$ if and only if $I_u \cap I_v \neq \emptyset$, 
where $I_u$ and $I_v$ denote
the intervals associated to $u$ and $v$, respectively. Such a graph is
said to admit an \emph{interval representation}. A graph is a
\emph{proper interval graph} if it admits an interval representation
such that $I_u \not\subset I_v$ for every $u,v \in V$. In other words,
no interval strictly contains another interval.\\ 
We will make use of the two following characterizations of proper
interval graphs to design our kernelization algorithm.
		
% characterization of proper interval graphs

\begin{theorem}[Forbidden subgraphs~\cite{Weg67}]
\label{thm:pic} 
A graph is a proper interval graph if and only if it does not contain
any $\{hole,claw,net,3$-$sun\}$ as an induced subgraph (see
Figure~\ref{fig:forbidden}).
\end{theorem}
		
\begin{figure}[ht]
\centerline{
\includegraphics[scale=0.85]{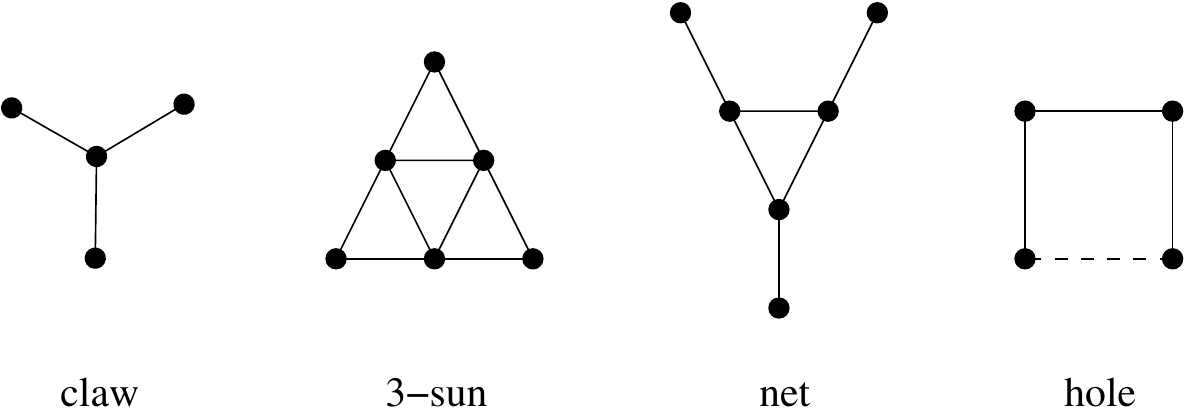}}
\caption{The forbidden induced subgraphs of proper interval graphs. A
  \emph{hole} is an induced cycle of length at least
  $4$. 
\label{fig:forbidden}}
\end{figure}

The claw graph is the bipartite graph $K_{1,3}$. Denoting the bipartition
by $(\{c\},\{l_1,l_2,l_3\})$, we call $c$ the \emph{center}
and $\{l_1,l_2,l_3\}$ the \emph{leaves} of the claw.

% umbrella order
\begin{theorem}[Umbrella property~\cite{LO93}]
\label{thm:umbrella} 
A graph is a proper interval graph if and only if its vertices admit
an ordering $\sigma$ (called \emph{umbrella ordering}) satisfying the
following property: given $v_iv_j \in E$ with $i < j$ then $v_iv_l, v_lv_j
\in E$ for every $i < l < j$ (see Figure~\ref{fig:umbrella}).
\end{theorem}

In the following, we associate an umbrella ordering $\sigma_G$ to any
proper interval graph $G = (V,E)$.  There are several things to
remark.  First, note that in an umbrella ordering $\sigma_G$ of a
graph $G$, every maximal set of true twins of $G$ is consecutive, and
that $\sigma_G$ is unique up to permutation of true twins of $G$.
Remark also that for any edge $uv$ with $u <_{\sigma_G} v$, the set
$\{w \in V\ :\ u\le_{\sigma_G} w\le_{\sigma_G} v\}$ is a clique of
$G$, and for every $i$ with $1\le i < l$, $(\{v_1,\dots
,v_i\},\{v_{i+1},\dots ,v_n\})$ is a join of $G$.\\
According to this ordering, we say that an edge $uv$ is
\emph{extremal} if there does not exist any edge $u'v'$ different of
$uv$ such that $u' \le_{\sigma_G} u$ and $v\le_{\sigma_G} v'$ (see
Figure~\ref{fig:umbrella}).

\begin{figure}[t]
\centerline{
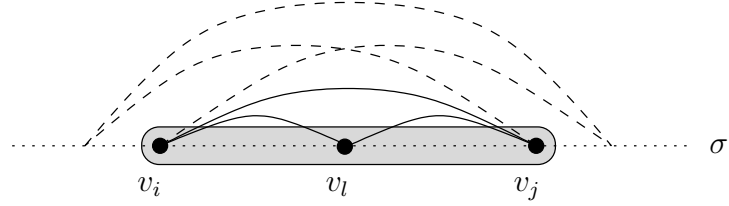}
\caption{Illustration of the umbrella property. 
The edge $v_iv_j$ is extremal.\label{fig:umbrella}
\protect\footnotemark}
\end{figure}
\footnotetext{In all the figures, (non-)edges between blocks stand for
all the possible (non-)edges between the vertices that lie in these blocks,
and the vertices within a gray box form a clique of the graph.} 
  
Let $G = (V,E)$ be an instance of \PIC{}. A \emph{completion} of $G$
is a set $F \subseteq (V \times V) \setminus E$ such that the graph $H
= (V, E \cup F)$ is a proper interval graph. In a slight abuse of
notation, we use $G + F$ to denote the graph $H$. A
\emph{$k$-completion} of $G$ is a completion such that $|F| \le
k$, and an \emph{optimal completion} $F$ is such that $|F|$ is
minimum. We say that $G = (V,E)$ is a \emph{positive} instance of
\PIC{} whenever it admits a $k$-completion.  We state a simple
observation that will be very useful for our kernelization algorithm.
		
\begin{observation}
\label{obs:extremal}
Let $G = (V,E)$ be a graph and $F$ be an optimal completion of
$G$. Given an umbrella ordering $\sigma$ of $G + F$, any extremal edge
of $\sigma$ is an edge of $G$.
\end{observation}
		
\begin{proof}
Assume that there exists an extremal edge $e$ in $\sigma$ that belongs
to $F$. By definition, $\sigma$ is still an umbrella ordering if we
remove the edge $e$ from $F$, contradicting the optimality of $F$.
\end{proof}

\subsection{Branches}

We now give the main definitions of this Section. The branches that we
will define correspond to some parts of the graph that already behave
like proper interval graphs. They are the parts of the graph that we
will reduce in order to obtain a kernelization algorithm.
		
% definition of the branches
\begin{definition}[$1$-branch]
\label{def:1branch}
Let $B \subseteq V$. We say that $B$ is a $1$-branch if the following
properties hold (see Figure~\ref{fig:1branch}):
\begin{enumerate}[(i)]
\item The graph $G[B]$ is a connected proper interval graph admitting an
  umbrella ordering $\sigma_{B} = b_1, \ldots ,b_{|B|}$ and,
\item The vertex set $V \setminus B$ can be partitioned into two sets
  $R$ and $C$ with: no edges between $B$ and $C$, every vertex in $R$ 
  has a neighbor in $B$, no edges between
  $\{b_1,\dots ,b_{l-1}\}$ and $R$ where $b_l$ is the neighbor of
  $b_{|B|}$ with minimal index in $\sigma_B$, and for every $l \le i < |B|$, we have
  $N_R(b_i) \subseteq N_R(b_{i+1})$.
\end{enumerate}
\end{definition} 

We denote by $B_1$ the set of vertices $\{v \in V:\ b_l \le_{\sigma_B}
v \le_{\sigma_B} b_{|B|}\}$, which is a clique (because $b_l$ is a
neighbor of $b_{|B|}$). We call $B_1$ the \emph{attachment clique} of
$B$, and use $B^R$ to denote $B \setminus B_1$.

\begin{figure}[ht]
\centerline{
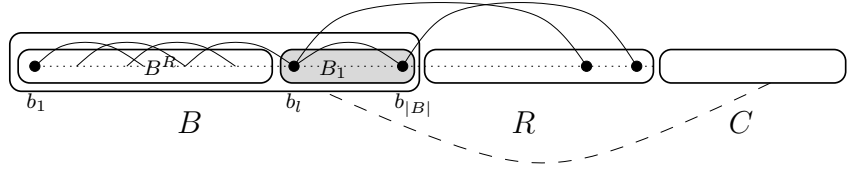}
\caption{A $1$-branch of a graph $G = (V, E)$. The vertices of $B$ are
  ordered according to the umbrella ordering
  $\sigma_B$. \label{fig:1branch}}
\end{figure}
		
\begin{definition}[$2$-branch]
\label{def:2branch}
Let $B \subseteq V$. We say that $B$ is a $2$-branch if the following
properties hold (see Figure~\ref{fig:2branch}):
\begin{enumerate}[(i)]
\item The graph $G[B]$ is a connected proper interval graph admitting an umbrella
  ordering $\sigma_{B} = b_1, \ldots ,b_{|B|}$ and,
\item The vertex set $V \setminus B$ can be partitioned into sets
  $L,R$ and $C$ with:
  \begin{itemize} 
    \item no edges between $B$ and $C$,
    \item every vertex in $L$ (resp. $R$) has a neighbor in $B$,
    \item no edges between $\{b_1,\dots ,b_{l-1}\}$ and $R$ where
      $b_l$ is the neighbor of $b_{|B|}$ with minimal index in
      $\sigma_B$,
    \item no edges between $\{b_{l'+1},\dots ,b_{|B|}\}$ and $L$ where
      $b_{l'}$ is the neighbor of $b_{1}$ with maximal index in
      $\sigma_B$ and,
    \item $N_R(b_i) \subseteq N_R(b_{i+1})$ for every $l \le i < |B|$
      and $N_L(b_{i+1}) \subseteq N_L(b_i)$ for every $1 \le i < l'$.
  \end{itemize}
\end{enumerate}
\end{definition}

Again, we denote by $B_1$ (resp. $B_2$) the set of vertices $\{v \in
V:\ b_1 \le_{\sigma_B} v \le_{\sigma_B} b_{l'}\}$ (resp. $\{v \in
V:\ b_l \le_{\sigma_B} v \le_{\sigma_B} b_{|B|}\}$).  We call $B_1$
and $B_2$ the \emph{attachment cliques} of $B$, and use $B^R$ to
denote $B \setminus (B_1 \cup B_2)$. Observe that the cases where $L =
\emptyset$ or $R = \emptyset$ are possible, and correspond to the
definition of a $1$-branch. Finally, when $B^R=\emptyset$, it is
possible that a vertex of $L$ or $R$ is adjacent to all the vertices
of $B$. In this case, we will denote by $N$ the set of vertices that
are adjacent to every vertex of $B$, remove them from $R$ and $L$ and
abusively still denote by $L$ (resp. $R$) the set $L\setminus N$
(resp. $R\setminus N$). We will precise when we need to use the set
$N$.\\

\begin{figure}[ht]
\centerline{
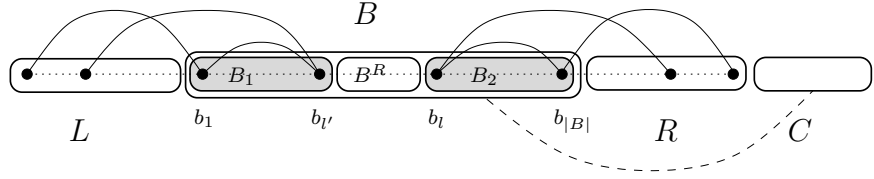}
\caption{A $2$-branch of a graph $G = (V, E)$. The vertices of $B$ are
  ordered according to the umbrella ordering
  $\sigma_B$. \label{fig:2branch}}
\end{figure}
		
In both cases, in a 1- or 2-branch, whenever the proper interval graph $G[B]$ is a
\emph{clique}, we say that $B$ is a \emph{$K$-join}. Observe that, in a 
1- or 2-branch $B$, for
any extremal edge $uv$ in $\sigma_B$, the set of vertices $\{w \in
V:\ u \le_{\sigma_B} w \le_{\sigma_B} v\}$ defines a $K$-join.  In
particular, this means that a branch can be decomposed into a sequence
of $K$-joins. Observe however that the decomposition is not unique: for
instance, the $K$-joins corresponding to all the extremal edges of $\sigma_B$ are 
not disjoint. We will precise in
Section~\ref{section:cutting_the_2branch}, when we will reduce the size
of 2-branches, how to fix a decomposition. Finally, we say that a
$K$-join is \emph{clean} whenever its vertices are not contained in any claw or
$4$-cycle. Remark that a subset of a $K$-join (resp. clean $K$-join) is
also a $K$-join (resp. clean $K$-join).

% ======================================================================================================================== %

\section{Kernel for \sc{\PIC}}
\label{sec:pic}

The basic idea of our kernelization algorithm is to detect the large
enough branches and then to reduce them. This section details the
rules we use for that.

\subsection{Reduction rules}

\subsubsection{Basic rules}

We say that a rule is \emph{safe} if when it is applied to an instance
$(G,k)$ of the problem, $(G,k)$ admits a $k$-completion iff the instance 
$(G',k')$ reduced by the rule admits a $k'$-completion.

The first reduction rule gets rid of connected components that are
already proper interval graphs. This rule is trivially safe 
and can be
applied in $O(n + m)$ time using any recognition algorithm for proper
interval graphs~\cite{Corn04}. 
		
% connected components rule
\begin{polyrule}[Connected components]
\label{rule:cc}
Remove any connected component of $G$ that is a proper interval graph.
\end{polyrule}

The following reduction rule can be applied since proper interval
graphs are closed under true twin addition and induced subgraphs. For
a class of graphs satisfying these two properties, we know that this
rule is safe~\cite{BPP10} (roughly speaking, we edit all the large set
of true twins in the same way).
		
% true twins rule
\begin{polyrule}[True twins~\cite{BPP10}]
\label{rule:twins}
Let $T$ be a set of true twins in $G$ such that $|T| > k$. Remove $|T|
- (k + 1)$ arbitrary vertices from $T$.
\end{polyrule}
		
We also use the classical \emph{sunflower} rule, allowing to identify
a set of edges that must be added in any optimal completion.
		
\begin{polyrule}[Sunflower]
\label{rule:sunflower}
Let $\mathcal{S} = \{C_1, \ldots, C_m\}$, $m > k$ be a set of claws
having  two leaves $u,v$ in common but distinct third leaves. Add $uv$ to $F$
and decrease $k$ by $1$.\\
Let $\mathcal{S} = \{C_1, \ldots, C_m\}$, $m > k$ be a set of
distinct $4$-cycles having a non-edge $uv$ in common. Add $uv$ to $F$
and decrease $k$ by $1$.
\end{polyrule}
		
% sunflower is safe
\begin{lemma}
\label{lem:sunflower}
Rule~\ref{rule:sunflower} is safe and can be carried out in polynomial time.
\end{lemma}
		
\begin{proof}
We only prove the first rule. The second rule can be proved similarly.
Let $F$ be a $k$-completion of $G$ and assume that $F$ does not
contain $(u,v)$. Since any two claws in $\mathcal{S}$ only share
$(u,v)$ as a common non-edge, $F$ must contain one edge for every
$C_i$, $1 \le i \le m$.  Since $m > k$, we have $|F| > k$, which
cannot be. Observe that a sunflower can be found in polynomial time
once we have enumerated all the claws and $4$-cycles of a graph, which
can clearly be done in $O(n^4)$.
\end{proof}

\subsubsection{Extracting a clean $K$-join from a $K$-join}
\label{section:extraction}	

Now, we want to reduce the size of the 'simplest' branches, namely the
$K$-joins.  More precisely, in the next subsection we will bound the
number of vertices in a clean $K$-join (whose vertices are not contain
in any claw or $4$-cycle), and so, we first indicate how to extract a
clean $K$-join from a $K$-join.

% few vertices contained in a claw and $4$-cycles in a positive instance
\begin{lemma}
\label{lem:claws_and_$4$-cycles}
Let $G = (V,E)$ be a positive instance of \PIC{} on which
Rule~\ref{rule:sunflower} has been applied. There are at most
$k^2$ claws with distinct sets of leaves, and at most
$k^2+2k$ vertices of $G$ are leaves of claw. Furthermore, there are at most 
$2k^2+2k$ vertices of $G$ that are vertices of a $4$-cycle.
\end{lemma}

\begin{proof}
As $G$ is a positive instance of \PIC{}, every claw or $4$-cycle of $G$
has a non-edge that will be completed and then is an edge of $F$.  Let
$xy$ be an edge of $F$. As we have applied Rule~\ref{rule:sunflower}
on $G$, there are at most $k$ vertices in $G$ that form the three
leaves of a claw with $x$ and $y$. So, at most $(k+2)k$ vertices of $G$
are leaves of claws.  Similarly, there are at most $k$ non-edges of $G$,
implying at most $2k$ vertices, that form a $4$-cycle with $x$ and $y$.
So, at most $(2k+2)k$ vertices of $G$ are in a $4$-cycle.
\end{proof}
	
\begin{lemma}
\label{lem:center_of_claws}
Let $G = (V,E)$ be a positive instance of \PIC{} on which
Rule~\ref{rule:twins} and Rule~\ref{rule:sunflower} have been applied
and $B$ be a $K$-join of $G$.  There are at most $k^3+4k^2+5k+1$
vertices of $B$ that belong to a claw or a $4$-cycle.
\end{lemma}

\begin{proof}
By Lemma~\ref{lem:claws_and_$4$-cycles}, there are at most $3k^2+4k$
vertices of $B$ that are leaves of a claw or in $4$-cycles. We remove
these vertices from $B$ and denote $B'$ the set of remaining vertices,
which forms a $K$-join. Now, we remove from $B'$ all the vertices that
do not belong to any claw and contract all the true twins in the
remaining vertices. As Rule~\ref{rule:twins} has been applied on $B$,
every contracted set has size at most $k+1$. We denote by $B''$ the
obtained set which can be seen as a subset of $B$ and then, $B''$ is
also a $K$-join of $G$. Remark that every vertex of $B''$ is the
center of a claw.  We consider an umbrella ordering $b_1,\dots ,b_l$
of $B''$. We will find a set of $l-1$ claws with distinct sets of
leaves, which will bound $l$ by $k^2+1$, by
Lemma~\ref{lem:claws_and_$4$-cycles}.  As, for all $i=1, \dots ,l-1$,
$b_i$ and $b_{i+1}$ are not true twins, there exists $c_i$ such that
$b_ic_i\in E$ and $b_{i+1}c_i\notin E$ or $b_ic_i\notin E$ and
$b_{i+1}c_i\in E$.  As $B''$ is $K$-join, by definition, all the
$c_i$'s are distinct. Now, for every $i=1,\dots ,l-1$, we will find a
claw containing $c_i$ as leaf. Assume that $b_ic_i\notin E$ and
$b_{i+1}c_i\in E$. As $b_{i+1}$ is the center of a claw, there exists
a set $\{x,y,z\}$ which is an independent set and is fully adjacent to
$b_{i+1}$. If $c_i\in \{x,y,z\}$, we are done. Assume this is not the
case. This means that $b_i$ is adjacent to any vertex of $\{x,y,z\}$
(otherwise one of this vertex would be adjacent to $b_{i+1}$ and not
to $b_i$, and we choose it to be $c_i$).  Now, if two elements of this
set, say $x$ and $y$, are adjacent to $c_i$, then $\{x,c_i,y,b_i\}$
forms a $4$-cycle that contains $b_i$, which is not possible. So, at
least two elements among $\{x,y,z\}$, say $x$ and $y$, are not
adjacent to $c_i$ and then, we find the claw $\{b_{i+1},x,y,c_i\}$ of
center $b_{i+1}$ that contains $c_i$. In the case where $b_ic_i\in E$
and $b_{i+1}c_i\notin E$, we proceed similarly by exchanging the role
of $b_i$ and $b_{i+1}$ and find also a claw containing $c_i$.
Finally, all the considered claws have distinct sets of leaves and
there are at most $k^2$ such claws by
Lemma~\ref{lem:claws_and_$4$-cycles}.  What means that $B''$ has size
at most $k^2+1$ and $B'$ at most $(k+1)(k^2+1)$.  As we removed
at most $3k^2+4k$ vertices of $B$ that could be leaves of claws or
contain in $4$-cycles, we obtain $k^3+4k^2+5k+1$ vertices of $B$ that
are possibly in claws or $4$-cycles.
\end{proof}

Since any subset of a $K$-join forms a $K$-join, Lemma~\ref{lem:center_of_claws} 
implies that it is possible to remove a set of at most $k^3+4k^2+5k+1$
vertices from any $K$-join to obtain a clean $K$-join.

\subsubsection{Bounding the size of the $K$-joins}

Now, we set a rule that will bound the number of vertices in 
a clean $K$-join, once applied. Although quite technical to prove, this
rule is the core tool of our process of kernelization.

\begin{polyrule}[$K$-join]
\label{rule:$K$-join}
Let $B$ be a clean $K$-join of size at least $2k+2$. Let $B_L$ be the $k +
1$ \emph{first vertices} of $B$, $B_R$ be its $k + 1$ \emph{last
  vertices} and $M = B \setminus (B_R \cup B_L)$. Remove the set of
vertices $M$ from $G$.
\end{polyrule}
		
% $K$-join rule is safe
\begin{lemma}
\label{lem:$K$-join}
Rule~\ref{rule:$K$-join} is safe.
\end{lemma}

\begin{proof}
Let $G' = G \setminus M$. Observe that the restriction to $G'$ of any
$k$-completion of $G$ is a $k$-completion of $G'$, since proper
interval graphs are closed under induced subgraphs. So, let $F$ be a
$k$-completion for $G'$. We denote by $H = G' + F$ the resulting
proper interval graph and by $\sigma_H=b_1,\dots ,b_{|H|}$ an umbrella
ordering of $H$. We prove that we can insert the vertices of
$M$ into $\sigma_H$ and modify it if necessary, to obtain an umbrella
ordering for $G$ without adding any edge (in fact, some edges of $F$
might even be deleted during the process). This will imply that $G$
admits a $k$-completion as well. To see this, we need the following
structural description of $G$. As explained before, we denote by $N$
the set $\cap_{b \in B} N_G(b) \setminus B$, and abusively still
denote by $L$ (resp. $R$) the set $L\setminus N$ (resp. $R\setminus
N$) (see Figure~\ref{fig:struct}).
			
\begin{claim}
\label{claim:struct}
The sets $L$ and $R$ are cliques of $G$.
\end{claim}
			
\emph{Proof.}  We prove that $R$ is a clique in $G$. The proof for $L$
uses similar arguments. No vertex of $R$ is a neighbor of $b_1$,
otherwise such a vertex must be adjacent to every vertex of $B$ and
then stand in $N$.  So, if $R$ contains two vertices $u,v$ such that
$uv \notin E$, we form the claw $\{b_{|B|},b_1,u,v\}$ of center $b_{|B|}$,
contradicting the fact that $B$ is clean.  \hfill $\diamond$ \\
			
The following observation comes from the definition of a $K$-join.
			
\begin{observation}
\label{obs:$K$-join}
Given any vertex $r \in R$, if $N_B(r) \cap B_L \neq \emptyset$ holds
then $M \subseteq N_B(r)$. Similarly, given any vertex $l \in L$, if
$N_B(l) \cap B_R \neq \emptyset$ holds then $M \subseteq N_B(l)$.
\end{observation}

\begin{figure}[ht]
\centerline{
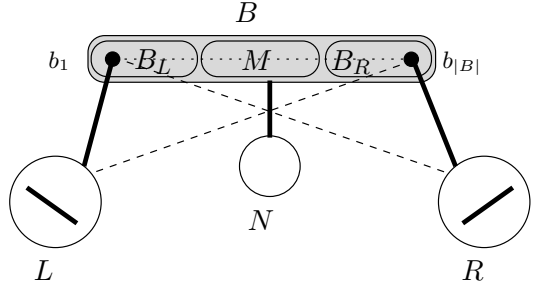}
\caption{The structure of the $K$-join $B$. \label{fig:struct}}
\end{figure}

We use these facts to prove that an umbrella ordering can be obtained
for $G$ by inserting the vertices of $M$ into $\sigma_H$. Let $b_f$
and $b_l$ be respectively the first and last vertex of $B \setminus M$
appearing in $\sigma_H$. We let $B_H$ denote the set $\{u \in
V(H)\ :\ b_f \le_{\sigma_H} u \le_{\sigma_H} b_l\}$.  Observe that
$B_H$ is a clique in $H$ since $b_fb_l \in E(G)$ and that $B\setminus
M \subseteq B_H$. Now, we modify $\sigma_H$ by ordering the true twins
in $H$ according to their neighborhood in $M$: if $x$ and $y$ are true
twins in $H$, are consecutive in $\sigma_H$, verify
$x<_{\sigma_H}y<_{\sigma_H}b_f$ and $N_M(y)\subset N_M(x)$, then we
exchange $x$ and $y$ in $\sigma_H$. This process stops when the
considered true twins are ordered following the join between $\{u \in
V(H)\ :\ u <_{\sigma_H} b_f\}$ and $M$. We proceed similarly on the
right of $B_H$, i.e. for $x$ and $y$ consecutive twins with
$b_l<_{\sigma_H}x<_{\sigma_H}y$ and $N_M(x)\subset N_M(y)$.  The
obtained order is clearly an umbrella ordering too (in fact, we just
re-labeled some vertices in $\sigma_H$), and we abusively still
denote it by $\sigma_H$.
			
\begin{claim}
\label{claim:clique}
The set $B_H \cup \{m\}$ is a clique of $G$ for any $m \in M$, and 
consequently $B_H\cup M$ is a clique of $G$.
\end{claim}
			
\emph{Proof.}  Let $u$ be any vertex of $B_H$. We claim that $um \in E(G)$. 
Observe that if $u \in
B$ then the claim trivially holds. So assume $u \notin B$.  Recall
that $B_H$ is a clique in $H$. It follows that $u$ is adjacent to
every vertex of $B \setminus M$ in $H$. Since $B_L$ and $B_R$ both
contain $k + 1$ vertices, we have $N_G(u) \cap B_L \neq
\emptyset$ and $N_G(u) \cap B_R \neq \emptyset$. Hence, $u$ belongs to
$L \cup N \cup R$ and $um \in E(G)$ by Observation~\ref{obs:$K$-join}.
\hfill $\diamond$
			
\begin{claim}
\label{claim:umbrellam}
Let $m$ be any vertex of $M$ and $\sigma'_H$ be the ordering obtained
from $\sigma_H$ by removing $B_H$ and inserting $m$ to the position of
$B_H$. The ordering $\sigma'_H$ respects the umbrella property.
\end{claim}
			
\emph{Proof.}  Assume that $\sigma'_H$ does not respect the umbrella
property, i.e. that there exist (w.l.o.g.) two vertices $u$ and $v$ of
$H\setminus B_H$ such that either $(1)$ $u<_{\sigma'_H} v
<_{\sigma'_H} m$, $um \in E(H)$ and $uv\notin E(H)$ or $(2)$
$u<_{\sigma'_H} m <_{\sigma'_H} v$, $um\notin E(H)$ and $uv\in E(H)$
or $(3)$ $u<_{\sigma'_H} v <_{\sigma'_H} m$, $um \in E(H)$ and $vm
\notin E(H)$.  First, assume that $(1)$ holds. Since $uv \notin E(H)$
and $\sigma_H$ is an umbrella ordering, $uw \notin E(H)$ for any $w
\in B_H$, and hence $uw \notin E(G)$.  This means that $B_L \cap
N_G(u) = \emptyset$ and $B_R \cap N_G(u) = \emptyset$, which is
impossible since $um \in E(G)$. Then, assume that $(2)$ holds. Since
$uv\in E(H)$ and $\sigma_H$ is an umbrella ordering, $B_H\subseteq
N_H(u)$, and in particular $B_L$ and $B_R$ are included in
$N_H(u)$. As $|B_L|=|B_R|=k+1$, we know that $N_G(u)\cap B_L\neq
\emptyset$ and $N_G(u)\cap B_R\neq \emptyset$, but then, Observation~\ref{obs:$K$-join}
implies that $um\in E(G)$.
So, $(3)$ holds, and we choose the first $u$ satisfying this property 
according to the order given by $\sigma'_H$. So we have $wm \notin
E(G)$ for any $w <_{\sigma'_H} u$. Similarly, we choose $v$ to be the
first vertex after $u$ satisfying $vm \notin E(G)$. Since
$um \in E(G)$, we know that $u$ belongs to $L \cup N \cup R$. 
Moreover, since $vm \notin E(G)$, $v \in C \cup L \cup
R$.  There are several cases to consider:
				
\begin{enumerate}[(i)]
\item $u \in N$: in this case we know that $B \subseteq N_G(u)$, and
  in particular that $ub_l \in E(G)$.  Since $\sigma_H$ is an umbrella
  ordering for $H$, it follows that $vb_l \in E(H)$ and $B_H\subseteq
  N_H(v)$. Since $|B_L| = |B_R| = k + 1$, we know that $N_G(v) \cap
  B_L \neq \emptyset$ and $N_G(v) \cap B_R \neq \emptyset$.  But, then
  Observation~\ref{obs:$K$-join} implies that $vm \in E(G)$.

\item $u \in R,\ v \notin R$: since $um \in E(G)$, $B_R \subseteq
  N_G(u)$. Let $b \in B_R$ be the vertex such that $B_R \subseteq \{w
  \in V\ :\ u <_{\sigma_H} w \le_{\sigma_H} b\}$. Since $ub \in E(G)$,
  this means that $B_R \subseteq N_H(v)$. Now, since $|B_R| = k + 1$,
  it follows that $N_G(v) \cap B_R \neq \emptyset$. 
  Observation~\ref{obs:$K$-join} allows us to conclude that $vm \in
  E(G)$.
 
\item $u, v \in R$: in this case, $uv \in E(G)$ by
  Claim~\ref{claim:clique} but $u$ and $v$ are not true twins in $H$
  (otherwise $v$ would be placed before $u$ in $\sigma_H$ due to the
  modification we have applied to $\sigma_H$). This means that there
  exists a vertex $w \in V(H)$ that distinguishes $u$ from $v$ in
  $H$.\\ Assume first that $w <_{\sigma_H} u$ and $uw \in E(H), vw
  \notin E(H)$. We choose the first $w$ satisfying this according to
  the order given by $\sigma_H$. There are two cases to
  consider. First, if $uw \in E(G)$, then since $wm \notin E(G)$ for
  any $w <_{\sigma_H} u$ by the choice of $u$, $\{u,v,w,m\}$ is a claw
  in $G$ containing a vertex of $B$ (see Figure~\ref{fig:umbrellam}
  $(a)$ ignoring the vertex $u'$), which cannot be. So assume $uw \in
  F$. By Observation~\ref{obs:extremal}, $uw$ is not an extremal edge
  of $\sigma_H$.  By the choice of $w$ and since $vw \notin E(H)$,
  there exists $u'$ with $u <_{\sigma_H} u' <_{\sigma_H} v$ such that
  $u'w$ is an extremal edge of $\sigma_H$ (and hence belongs to
  $E(G)$, see Figure~\ref{fig:umbrellam} $(a)$). Now, by the choice of
  $v$ we have $u'm \in E(G)$ and hence $u' \in N \cup R \cup L$.
  Observe that $u'v \notin E(G)$: otherwise $\{u',v,w,m\}$ would form
  a claw in $G$. Since $R$ is a clique of $G$, it follows that $u' \in
  L \cup N$. Moreover, since $u'm \in E(G)$, $B_L \subseteq
  N_G(u')$. We conclude like in configuration $(ii)$ that $v$ should
  be adjacent to a vertex of $B_L$ and hence to $m$.\\ Hence we can
  assume that all the vertices that distinguish $u$ and $v$ are after
  $u$ in $\sigma_H$ and that $uw'' \in E(H)$ implies $vw'' \in E(H)$
  for any $w'' <_{\sigma_H} u$. Now, suppose that there exists $w\in
  H$ such that $b_l <_{\sigma_H} w$ and $uw \notin E(H),\ vw \in
  E(H)$. In particular, this means that $B_L \subseteq N_H(v)$.  Since
  $|B_L| = k + 1$ we have $N_G(v) \cap B_L \neq \emptyset$, implying
  $vm \in E(G)$ by Observation~\ref{obs:$K$-join}. Assume now that
  there exists a vertex $w$ which distinguishes $u$ and $v$ with $v
  <_{\sigma_H} w <_{\sigma_H} b_f$.  In this case, since $uw \notin
  E(H)$, $B \cap N_H(u) = \emptyset$ holds and hence $B \cap N_G(u) =
  \emptyset$, which cannot be since $u \in R$. Finally, assume that
  there is $w \in B_H$ with $wu\notin E(H)$ and $wv\in E(H)$.  Recall
  that $wm\in E(G)$ as $B_H\cup \{m\}$ is a clique by
  Claim~\ref{claim:clique}. We choose $w$ in $B_H$ distinguishing $u$
  and $v$ to be the last according to the order given by $\sigma_H$
  (i.e. $vw' \notin E(H)$ for any $w <_{\sigma_H} w'$, see
  Figure~\ref{fig:umbrellam} $(b)$, ignoring the vertex $u'$).

\begin{figure}[ht]
\centerline{
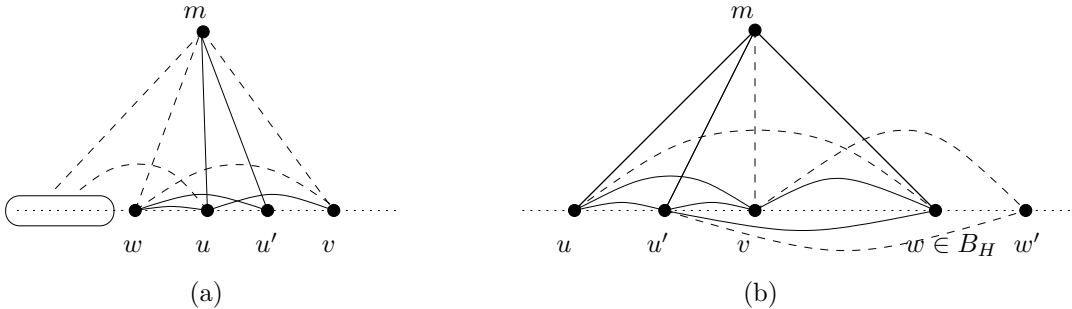}
\caption{$(a)$ $u$ and $v$ are distinguished by some vertex $w
  <_{\sigma_H} u$; $(b)$ $u$ and $v$ are distinguished by a vertex $w
  \in B_H$.\label{fig:umbrellam}}
  \end{figure}

If $vw \in E(G)$ then $\{u,m,w,v\}$ is a $4$-cycle in $G$ containing a
vertex of $B$, which cannot be.  Hence $vw \in F$ and by the choice of
$w$, there exists $u' \in V(H)$ such that $u <_{\sigma_H} u'
<_{\sigma_H} v$ and $u'w$ is an extremal edge (and then belongs to
$E(G)$). By the choice of $v$ we know that $u'm \in E(G)$. Moreover,
by the choice of $w$, observe that $u'$ and $v$ are true twins in $H$
(if a vertex $s$ distinguishes $u'$ and $v$ in $H$, $s$ cannot be
before $u$, since otherwise $s$ would distinguishes $u$ and $v$, not
between $u$ and $w$ because it would be adjacent to $u'$ and $v$, and
not after $w$, by choice of $w$). This leads to a contradiction since
we assumed that $N_M(x) \subseteq N_M(y)$ for any true twins $x$ and
$y$ with $x <_{\sigma_H} y <_{\sigma_H} b_f$.
\end{enumerate}
				
The cases where $u \in L$ are similar, what concludes the proof of
Claim~\ref{claim:umbrellam} \hfill $\diamond$ \\
			
\begin{claim}
\label{claim:properm}
Let $m \in M$. Then $m$ can be added to the graph $H$ while preserving
an umbrella ordering.
\end{claim}
		
\emph{Proof.}  Let $m \in M$ and $v_i$ (resp. $v_j$) be the vertex
with minimal (resp. maximal) index in $\sigma_H$ such that $v_im \in
E(G)$ (resp. $v_jm \in E(G)$). By definition, we have $v_{i-1}m,
v_{j+1}m \notin E(G)$ and through Claim~\ref{claim:umbrellam}, we know
that $N_H(m) = \{w \in V\ :\ v_i \le_{\sigma_H} w \le_{\sigma_H}
v_j\}$.  Moreover, since $B_H \cup M$ is a clique by
Claim~\ref{claim:clique}, it follows that $v_{i-1} <_{\sigma_H} b_f$
and $b_l <_{\sigma_H} v_{j+1}$.  Hence, by
Claim~\ref{claim:umbrellam}, we know that $v_{i-1}v_{j+1} \notin
E(G)$, otherwise the ordering $\sigma'_H$ defined in
Claim~\ref{claim:umbrellam} would not be an umbrella ordering.  The
situation is depicted in Figure~\ref{fig:x1x2} $(a)$.  For any vertex
$v \in N_H(m)$, let $N^-(v)$ (resp. $N^+(v)$) denote the set of
vertices $\{w\le_{\sigma_H} v_{i-1}\ :\ wv \in E(H)\}$ (resp. $\{w
\ge_{\sigma_H} v_{j+1}\ :\ wv \in E(H)\}$). Observe that for any
vertex $v\in N_H(m)$, if there exist two vertices $x \in N^-(v)$ and
$y \in N^+(v)$ such that $xv, yv \in E(G)$, then the set $\{v,x,y,m\}$
defines a claw containing $m$ in $G$, which cannot be. We now consider
$b_{v_{i-1}}$ the neighbor of $v_{i-1}$ with maximal index in
$\sigma_H$.  Similarly we let $b_{v_{j+1}}$ be the neighbor of
$v_{j+1}$ with minimal index in $\sigma_H$. Since $v_{i-1}v_{j+1}
\notin E(G)$, we have $b_{v_{i-1}}, b_{v_{j+1}} \in N_H(m)$.  
We study
the behavior of $b_{v_{i-1}}$ and $b_{v_{j+1}}$ in order to conclude.
	
Assume first that $b_{v_{j+1}} <_{\sigma_H} b_{v_{i-1}}$. 
Let $X$
be the set of vertices $\{w \in V\ :\ b_{v_{j+1}} \le_{\sigma_H} w
\le_{\sigma_H} b_{v_{i-1}}\}$. Remark that we have $b_{v_{i-1}}\leq_{\sigma_H}b_l$ 
and  $b_f \leq_{\sigma_H} b_{v_{j+1}}$, otherwise for instance, if 
we have $b_{v_{i-1}}>_{\sigma_H}b_l$, then $B_H\subseteq N_H(v_{i-1})$ implying,
as usual, that $v_{i-1}m\in E(G)$ which is not.
So, we know that $X\subseteq B_H$. Then, let $X_1 \subseteq X$ be the
set of vertices $x \in X$ such that there exists $w \in N^+(x)$ with
$xw \in E(G)$ and $X_2 = X \setminus X_1$. Let $x \in X_1$: observe
that by construction $xw' \in F$ for any $w' \in N^-(x)$.  Similarly,
given $x \in X_2$, $xw'' \in F$ for any $w'' \in N^+(x)$.  We now
reorder the vertices of $X$ as follows: we first put the vertices from
$X_2$ and then the vertices from $X_1$, preserving the order induced
by $\sigma_H$ for both sets. Moreover, we remove from $E(H)$ all edges
between $X_1$ and $N^-(X_1)$ and between $X_2$ and $N^+(X_2)$. Recall
that such edges have to belong to $F$.  We claim that
inserting $m$ \emph{between} $X_2$ and $X_1$ yields an umbrella
ordering (see Figure ~\ref{fig:x1x2} $b$). Indeed, by
Claim~\ref{claim:umbrellam}, we know that the umbrella ordering is
preserved between $m$ and the vertices of $H \setminus B_H$.

\begin{figure}[ht]
\centerline{
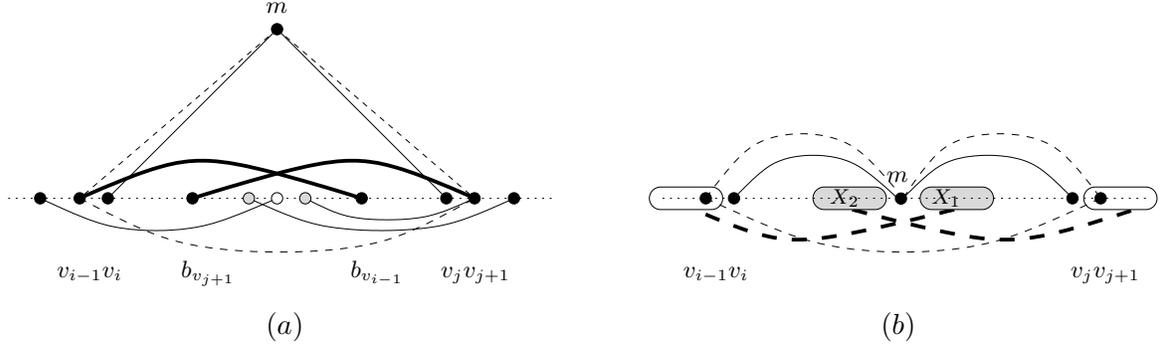}
\caption{Illustration of the reordering applied to $\sigma_H$. The
  thin edges stand for edges of $G$. On the left, the gray vertices
  represent vertices of $X_1$ while the white
  vertex is a vertex of $X_2$.\label{fig:x1x2}}
\end{figure}

Now, remark that there is no edge between $X_1$ and $\{w\in V\ :\ w
\le_{\sigma_H} v_{i-1}\}$, that there is no edge between $X_2$ and
$\{w\in V\ :\ w \ge_{\sigma_H} v_{j+1}\}$), that there are still all
the edges between $N_H(m)$ and $X_1\cup X_2$ and that the edges between
$X_1$ and $\{w\in V\ :\  w \ge_{\sigma_H} v_{j+1}\}$ and  the edges between
$X_2$ and $\{w\in V\ :\  w \le_{\sigma_H} v_{i-1}\}$ are unchanged.
So, it follows that the new
ordering respects the umbrella property, and we are done.
	
Next, assume that $b_{v_{i-1}} <_{\sigma_H} b_{v_{j+1}}$. We let $b_{v_i}$
(resp. $b_{v_j}$) be the neighbor of $v_i$ (resp. $v_j$) with maximal
(resp. minimal) index in $N_H(m)$. Notice that $b_{v_{i-1}}
\le_{\sigma_H} b_{v_i}$ and $b_{v_{j}} \le_{\sigma_H} b_{v_{j+1}}$
(see Figure~\ref{fig:notx1x2}). Two cases may occur:
	
\begin{enumerate}[(i)]
\item First, assume that $b_{v_{i}} <_{\sigma_H} b_{v_{j}}$, case
  depicted in Figure~\ref{fig:notx1x2} $(a)$. In particular, this
  means that $v_iv_j \notin E(G)$.  If $b_{v_i}$ and $b_{v_j}$ are
  consecutive in $\sigma_H$, then inserting $m$ between $b_{v_i}$ and
  $b_{v_j}$ yields an umbrella ordering (since $b_{v_j}$
  (resp. $b_{v_i}$) does not have any neighbor before (resp. after)
  $v_i$ (resp. $v_j$) in $\sigma_H$). Now, if there exists $w \in V$
  such that $b_{v_i} <_{\sigma_H} w <_{\sigma_H} b_{v_j}$, then one
  can see that the set $\{m,v_i,w,v_j\}$ forms a claw containing $m$
  in $G$, which is impossible.
\item The second case to consider is when $b_{v_j} \le_{\sigma_H}
  b_{v_i}$. In such a case, one can see that $m$ and the vertices of
  $\{w \in V\ :\ b_{v_j} \le_{\sigma_H} w \le_{\sigma_H} b_{v_i}\}$
  are true twins in $H \cup \{m\}$, because their common neighborhood
  is exactly $\{w\in V\ :\ v_i\le_{\sigma_H} w\le_{\sigma_H}v_j\}$.
  Hence, inserting $m$ just before $b_{v_i}$ (or anywhere between
  $b_{v_i}$ and $b_{v_j}$ or just after $b_{v_j}$) yields an umbrella
  ordering.
\end{enumerate}
			
\begin{figure}[ht]
\centerline{
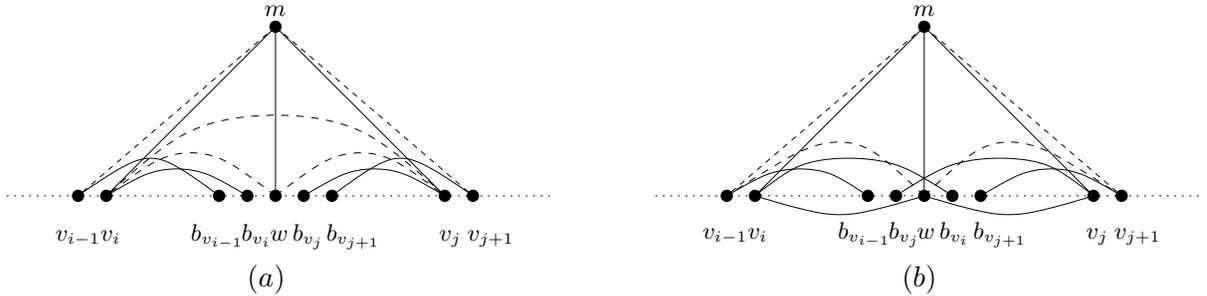}
\caption{The possible cases for $b_{v_{i-1}} <_{\sigma_H} b_{v_{j+1}}$. \label{fig:notx1x2}}
\end{figure}

\vspace{-0.5cm}
	
\hfill $\diamond$ \\
		
Since the proof of Claim~\ref{claim:properm} does not use the fact
that the vertices of $H$ do not belong to $M$, it follows that we can
iteratively insert the vertices of $M$ into $\sigma_H$, preserving an
umbrella ordering at each step. This concludes the proof of
Lemma~\ref{lem:$K$-join}.
\end{proof}

The complexity needed to compute Rule~\ref{rule:$K$-join} will be
discussed in the next section.  The following observation results from
the application of Rule~\ref{rule:$K$-join} and from
Section~\ref{section:extraction}.

\begin{observation}
\label{obs:clean$K$-join}
Let $G = (V,E)$ be a positive instance of \PIC{} reduced under
Rules~\ref{rule:twins} to \ref{rule:$K$-join}. Any $K$-join of $G$ has
size at most $k^3+4k^2+7k+3$.
\end{observation}

\begin{proof}
Let $B$ be any $K$-join of $G$, and assume $|B|>k^3+4k^2+7k+3$. 
By Lemma~\ref{lem:claws_and_$4$-cycles} we know that 
it is possible to extract a clean $K$-join from $B$ of size
at least $|B|-(k^3+4k^2+5k+1)>2(k+1)$ what is impossible after having 
applied Rule~\ref{rule:$K$-join}.
\end{proof}
		
\subsubsection{Cutting the $1$-branches}

We now turn our attention to branches of a graph $G = (V,E)$, proving
how they can be reduced.
%First, we prove a structural result on the $k$-completion.
		
\begin{lemma}
\label{lem:1branches}
Let $G = (V,E)$ be a connected graph and $B$ be a $1$-branch of $G$
associated with the umbrella ordering $\sigma_B$. Assume that $|B^R|
\ge 2k + 1$ and let $B_f$ be the $2k +1$ last vertices of $B^R$
according to $\sigma_B$. For any $k$-completion $F$ of $G$ into a
proper interval graph, there exists a $k$-completion $F'$ of $G$
with $F'\subseteq F$ and  a vertex $b \in B_f$ such that the
umbrella ordering of $G + F'$ preserves the order of the set $B_b = \{v
\in V\ :\ b_1 \le_{\sigma_B} v \le_{\sigma_B} b_{l'}\}$, where $l'$ is
the maximal index such that $bb_{l'} \in E(G)$. Moreover, the vertices of
$B_b$ are the first in an umbrella ordering of $G+F'$.
\end{lemma}
		
\begin{proof}
Let $F$ be any $k$-completion of $G$, $H = G + F$ and $\sigma_H$ be
the umbrella ordering of $H$. Since $|B _f| = 2k + 1$ and $|F|
\le k$, there exists a vertex $b \in B_f$ not incident to any
added edge of $F$.  We let $N_D$ be the set of neighbors of $b$ that are
\emph{after} $b$ in $\sigma_B$, $B'$ the set of vertices that are
\emph{before} $N_G[b]$ in $\sigma_B$, $B_b = B' \cup N_G[b]$ and $C = V \setminus B_b$
(see Figure~\ref{fig:1branchrule}).
			
\begin{claim}
\label{claim:useful}
$(i)$ $G[C]$ is a connected graph and\\
$(ii)$ Either $\forall u\in C\ b<_{\sigma_H}u$ holds or 
 $\forall u\in C\ u<_{\sigma_H}b$ holds.

\end{claim} 

\emph{Proof.} The first point follows from the fact that $G$ is
connected and that, by construction, $B_1 \subseteq C$ and $B_1$ is
connected. To see the second point, assume that there exist $u,v \in
C$ such that w.l.o.g. $u <_{\sigma_H} b <_{\sigma_H} v$. Since $G[C]$
is a connected graph, there exists a path between $u$ and $v$ in $G$
that avoids $N_G[b]$, which is equal to $N_H[b]$ since $b$ is not
incident to any edge of $F$. Hence there exist $u',v' \in C$ such that
$u' <_{\sigma_H} b <_{\sigma_H} v'$ and $u'v' \in E(G)$. Then, we have
$u'b, v'b \notin E(H)$, contradicting the fact that $\sigma_H$ is an
umbrella ordering for $H$.  \hfill $\diamond$ \\
			
In the following, we assume w.l.o.g. that $b <_{\sigma_H} u$ holds for
any $u \in C$.  We now consider the following ordering $\sigma$ of
$H$: we first put the set $B_b$ according to the order of $B$ and then
put the remaining vertices $C$ according to $\sigma_H$ (see
Figure~\ref{fig:1branchrule}).  We construct a completion $F'$ from
$F$ as follows: we remove from $F$ the edges with both extremities in
$B_b$, and remove all edges between $B_b \setminus N_D$ and $C$. In
other words, we set:
$$F' = F \setminus (F[B] \cup F[(B_b \setminus N_D) \times C])$$
Finally, we inductively remove from $F'$ any extremal edge of 
$\sigma$ that belongs to $F$, and abusively still call $F'$ 
the obtained edge set.
			
\begin{figure}[ht]
\centerline{
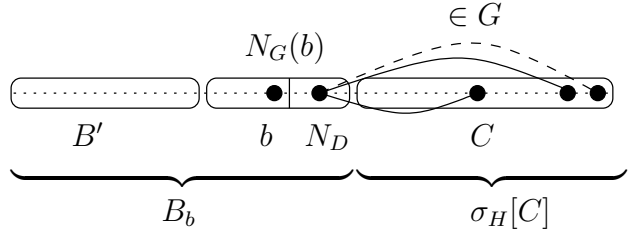}
\caption{The construction of the ordering $\sigma$ according to
  $\sigma_H$.\label{fig:1branchrule}}
\end{figure}
			
\begin{claim}
\label{claim:newopt}
The set $F'$ is a $k$-completion of $G$.% into a proper interval graph.
\end{claim}

\emph{Proof.}  We prove that $\sigma$ is an umbrella ordering of $H' =
G + F'$. Since $|F'| \le |F|$ by construction, the result will
follow. Assume this is not the case. By definition of $F'$, $H'[B_b]$
and $H'[C]$ induce proper interval graphs. This means that there
exists a set of vertices $S = \{u,v,w\}$, $u <_{\sigma} v <_{\sigma}
w$, intersecting both $B_b$ and $C$ and violating the umbrella
property. We either have $(1)$ $uw \in E, uv \notin E$ or $(2)$ $uw
\in E, vw \notin E$.  Since neither $F'$ nor $G$ contain an edge
between $B_b \setminus N_D$ and $C$, it follows that $S$ intersects
$N_D$ and $C$. We study the different cases:

\begin{enumerate}[(i)]
\item \label{item1:1branches} $(1)$ holds and $u \in N_D,\ v,w \in C$:
  since the edge set between $N_D$ and $C$ is the same in $H$ and
  $H'$, it follows that $uv \notin E(H)$. Since $\sigma_H$ is an
  umbrella ordering of $H$, we either have $v <_{\sigma_H} u
  <_{\sigma_H} w$ or $v <_{\sigma_H} w <_{\sigma_H} u$ (recall that
  $C$ is in the same order in both $\sigma$ and $\sigma_H$).  Now,
  recall that $b <_{\sigma_H} \{v,w\}$ by assumption. In particular,
  since $bu \in E(G)$, this implies in both cases that $\sigma_H$ is
  not an umbrella ordering, what leads to a contradiction.
\item $(1)$ \label{item2:1branches} holds and $u,v \in N_D,\ w \in C$:
  this case cannot happen since $N_D$ is a clique of $H'$.
\item $(2)$ \label{item3:1branches} holds and $u \in N_D,\ v,w \in C$:
  this case is similar to (\ref{item1:1branches}). Observe that we may assume 
  $uv \in E(H)$ (otherwise (\ref{item1:1branches}) holds). By construction $vw
  \notin E(H)$ and hence $v <_{\sigma_H} w <_{\sigma_H} u$ or $v
  <_{\sigma_H} u <_{\sigma_H} w$. The former case contradicts the fact
  that $\sigma_H$ is an umbrella ordering since $bu \in E(H)$. In the
  latter case, since $\sigma_H$ is an umbrella ordering this means
  that $bv \in E(H)$. Since $b$ is non affected vertex and $bv \notin
  E(G)$, this leads to a contradiction.
\item $(2)$ \label{item4:1branches} holds and $u,v \in N_D,\ w \in C$:
  first, if $uw \in E(G)$, then we have a contradiction since $N_C(u)
  \subseteq N_C(v)$. So, we have $uw \in F'$. By construction of $F'$,
  we know that $uw$ is not an extremal edge. Hence there exists an
  extremal edge (of $G$) containing $uw$, which is either $uw'$ with
  $w <_{\sigma} w'$ , $u'w$ with $u' <_{\sigma} u$ or $u'w'$ with $u'
  <_{\sigma} u <_{\sigma} w <_{\sigma} w'$. The three situation are
  depicted in Figure~\ref{fig:1branchrule1}.  In the first case, $vw' \in E(G)$
  (since $N_C(u) \subseteq N_C(v)$ in $G$) and hence we are in configuration
  (\ref{item1:1branches}) with vertex set $\{v,w,w'\}$.  In the second
  case, since $u'w \in E(G)$, we have a contradiction since $N_C(u')
  \subseteq N_C(v)$ in $G$ (observe that $u' \in B$ by construction).
  Finally, in the third case, $uw', vw' \in E(G)$ since $N_C(u')
  \subseteq N_C(u) \subseteq N_C(v)$ in $G$, and we are in configuration
  (\ref{item1:1branches}) with vertex set $\{v,w,w'\}$.
			
\begin{figure}[ht]
\centerline{
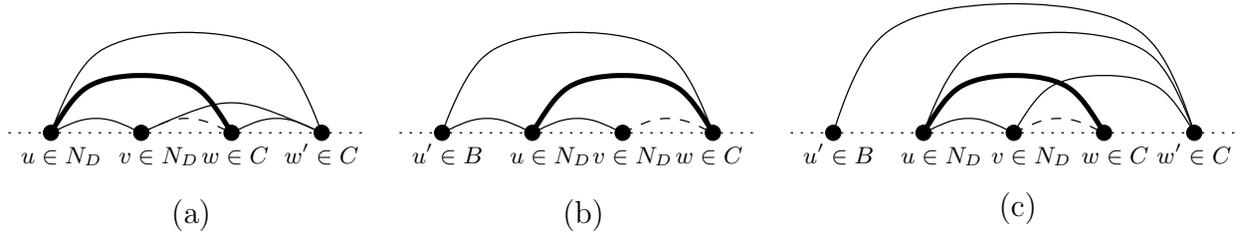}
\caption{Illustration of the different cases of configuration
  (\ref{item4:1branches}) (the bold edges belong to $F'$). \label{fig:1branchrule1}}
\end{figure}
\end{enumerate}

\vspace{-0.5cm}

\hfill $\diamond$ \\
			
Altogether, we proved that for any $k$-completion $F$, there
exists an umbrella ordering where the vertices of $B_b$ are ordered in
the same way than in the ordering of $B$ and stand at the beginning of this
ordering, what concludes the proof.
\end{proof}
		
% cutting the 1-branches rule
\begin{polyrule}[$1$-branches]
\label{rule:1branches}
Let $B$ be a $1$-branch such that $|B^ R| \ge 2k + 1$.  Remove
$B^R \setminus B_f$ from $G$, where $B_f$ denotes the $2k +1$ last
vertices of $B^R$.
\end{polyrule}
		
\begin{lemma}
\label{lem:1branchesrule}
Rule~\ref{rule:1branches} is safe.
\end{lemma}
		
\begin{proof}
Let $G' = G \setminus (B^R \setminus B_f)$ denote the reduced
graph. Observe that any $k$-completion of $G$ is a $k$-completion of
$G'$ since proper interval graphs are closed under induced
subgraphs. So let $F$ be a $k$-completion of $G'$. We denote by $H =
G' + F$ the resulting proper interval graph and let $\sigma_H$ be the
corresponding umbrella ordering. By Lemma~\ref{lem:1branches} we know
that there exists a vertex $b \in B_f$ such that the order of $B_b =
N_G[b] \cup \{v \in B_f\ :\ v <_{\sigma_B} N_G[b]\}$ is the same than in
$B$ and the vertices of $B_b$ are the first of $\sigma_H$. Since
$N_G(B^R \setminus B_f) \subseteq N_G[b]$, it follows that the
vertices of $B^R \setminus B_f$ can be inserted into $\sigma_H$ while
respecting the umbrella property. Hence $F$ is a $k$-completion for
$G$, implying the result.
\end{proof}
		
Here again, the time complexity needed to compute
Rule~\ref{rule:1branches} will be discussed in the next section.  The
following property of a reduced graph will be used to bound the size
of our kernel.
		
\begin{observation}
\label{obs:1branch}
Let $G = (V,E)$ be a positive instance of \PIC{} reduced under
Rules~\ref{rule:twins} to \ref{rule:1branches}. The $1$-branches of
$G$ contain at most $k^3+4k^2+9k+4$ vertices.
\end{observation}
		
\begin{proof}
Let $B$ be a $1$-branch of a graph $G = (V,E)$ reduced under
Rules~\ref{rule:twins} to \ref{rule:1branches}. Assume $|B| >
k^3+4k^2+9k+4$. Since $G$ is reduced under Rule~\ref{rule:$K$-join}, we
know by Observation~\ref{obs:clean$K$-join} that the attachment clique
$B_1$ of $B$, which is a $K$-join, contains at most $k^3+4k^2+7k+3$
vertices. This implies that $|B^R| > 2k+1$, which cannot be
since $G$ is reduced under Rule~\ref{rule:1branches}.
		\end{proof}

\subsubsection{Cutting the $2$-branches}
\label{section:cutting_the_2branch}

To obtain a rule reducing the 2-branches, we need to introduce a
particular decomposition of 2-branches into $K$-joins.  Let $B$ be a
2-branch with an umbrella ordering $\sigma_B=b_1,\dots ,b_{|B|}$.  As
usual, we denote by $B_1=b_1,\dots, b_{l'}$ its first attachment
clique and by $B_2=b_l,\dots ,b_{|B|}$ its second. The reversal of the
permutation $\sigma_B$ gives a second possibility to fix $B_1$ and
$B_2$. We fix one of these possibilities and define ${\cal B}$, the
\emph{$K$-join decomposition} of $B$.  The $K$-joins of $\cal B$ are
defined by $B'_i=b_{l_{i-1}+1},\dots ,b_{l_{i}}$ where $b_{l_{i}}$ is
the neighbor of $b_{l_{i-1}+1}$ with maximal index. The first $K$-join
of $\cal B$ is $B_1$ (so, $l_0=0$ and $l_1=l'$), and once $B'_{i-1}$
is defined, we set $B'_{i}$: if $b_{l_{i-1}+1} \in B_2$, then we
set $B'_i=b_{l_{i-1}+1},\dots ,b_{|B|}$, otherwise, we choose
$B'_i=b_{l_{i-1}+1},\dots ,b_{l_{i}}$ (see
Figure~\ref{fig:bijoin_dec}).
	
\begin{figure}[ht]
\centerline{
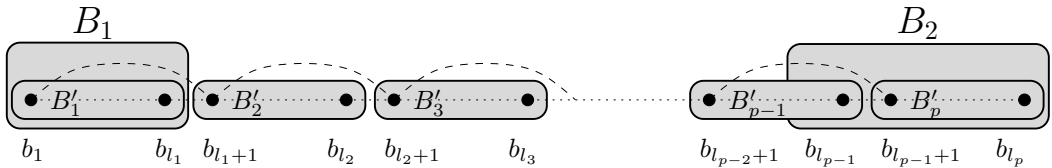}
\caption{The $K$-join decomposition. \label{fig:bijoin_dec}}
\end{figure}

Now, we can prove the next lemma, which bounds the number of $K$-joins 
in the $K$-join decomposition of a $2$-branch providing that some 
connectivity assumption holds.

\begin{lemma}
\label{lem:2branchescc}
Let $G = (V,E)$ be an instance of \PIC{} and $B$ be a $2$-branch
containing  $p \ge (k+4)$ $K$-joins in its $K$-join
decomposition. Assume the attachment cliques $B_1$ and $B_2$ of $B$
belong to the same connected component of $G[V \setminus B^R]$ . Then
there is no $k$-completion for $G$.
\end{lemma}
		
\begin{proof}
Let $B$ be a $2$-branch of an instance $G = (V,E)$ of \PIC{}
respecting the conditions of Lemma~\ref{lem:2branchescc}. Since $B_1$
and $B_2$ belong to the same connected component in $G[V \setminus
  B^R]$, let $\pi$ be a shortest path between $B_1$ and $B_2$ in $G[V
  \setminus B^R]$. As $B$ has $p\ge k+4\ge 3$ $K$-joins in its
decomposition, no vertex of $B_1$ is adjacent to a vertex f $B_2$ and
$\pi$ has length at least two.  We denote by $u \in B_1$ and $v \in
B_2$ the extremities of such a path. We now construct an induced path
$P_{uv}$ of length at least $p-1$ between $u$ and $v$ within $B$. To
do so, considering the $K$-join decomposition ${\cal B}=\{B'_1,\dots
,B'_p\}$ of $B$, we know that $u\in B'_1$ and that $v\in B'_{p-1}\cup
B'_p$.  We define $u_1=u$ and while $v\notin N[u_i]$, we choose
$u_{i+1}$ the neighbor of $u_i$ with maximum index in the umbrella
ordering of $B$.  In this case, we have $u_{i}\in \cup_{j=1}^{i}B'_j$
for every $1 \le i \le p$.  Indeed, the neighbor of a vertex of
$\cup_{j=1}^{i-1}B'_j$ with maximum index in the umbrella ordering of
$B$ is in $\cup_{j=1}^{i}B'_j$. Finally, when $v\in N[u_i]$, we just
choose $u_{i+1}=v$. So, the path $P_{uv}=u_1,\dots, u_l$ is an induced
path of length at least $p-1$, with $u_1=u$, $u_l=v$ and the only
vertices that could have neighbors in $G\setminus B$ are $u_1,
u_{l-1}$ and $u_l$ ($u_1\in B_1$, $u_l\in B_2$ and $u_{l-1}$ is
possibly in $B_2$). Using $\pi$, we can form an induced cycle of
length at least $p\ge k+4$ in $G$. Since at least $q - 3$ completions
are needed to triangulate any induced cycle of length
$q$~\cite{KST94}, it follows that there is no $k$-completion for $G$.
\end{proof}
		
The following observation is a straightforward implication of
Lemma~\ref{lem:2branchescc}.      
		
\begin{observation}
\label{obs:2branch_con}
Let $G = (V,E)$ be a connected positive instance of \PIC{}, reduced by
Rule~\ref{rule:$K$-join} and $B$ be a $2$-branch such that $G[V
  \setminus B^R]$ is connected. Then $B$ contains at most $k + 3$
$K$-joins in its $K$-join decomposition and hence at most $(k + 3)
(k^3+4k^2+5k+1)$ vertices.
\end{observation}
		
% cutting the 2-branches rule
\begin{polyrule}[$2$-branches]
\label{rule:2branches}
Let $G$ be a connected graph and $B$ be a $2$-branch such that $G[V
  \setminus B^ R]$ is not connected.  Assume that $|B^ R| \ge
4(k +1)$ and let $B'_1$ be the $2k + 1$ vertices after $B_1$ and
$B'_2$ the $2k+1$ vertices before $B_2$.  Remove $B \setminus (B_1
\cup B'_1 \cup B'_2 \cup B_2)$ from $G$.
\end{polyrule}
		
\begin{lemma}
\label{lem:2branchesrule}
Rule~\ref{rule:2branches} is safe.
\end{lemma}
		
\begin{proof}
As usual, we denote by $\sigma_b=b_1,\dots ,b_{|B|}$ the umbrella
ordering defined on $B$, with $B_1=\{b_1,\dots ,b_{l'}\}$ and
$B_2=\{b_l,\dots ,b_{|B|}\}$.  We partition $B^R$ into two sets
$B'=\{b_{l'+1},\dots ,b_i\}$ and $B''=\{b_{i+1},\dots ,b_{l-1}\}$ such
that $|B'| \ge |B''| \ge 2k + 1$. We now remove the edges $E(B',B'')$
between $B'$ and $B''$, obtaining two connected components of $G$,
$G_1$ and $G_2$. Observe that $B'$ defines a $1$-branch of $G_1$ with
attachment clique $B_1$ such that $B' \setminus B_1$ contains at least
$2k + 1$ vertices. Similarly $B''$ defines a $1$-branch of $G_2$ with
attachment clique $B_2$ such that $B'' \setminus B_2$ contains at
least $2k + 1$ vertices.  Hence Lemma~\ref{lem:1branches} can be
applied to both $G_1$ and $G_2$ and we continue as if
Rule~\ref{rule:1branches} has been applied to $G_1$ and $G_2$,
preserving exactly $2k + 1$ vertices $B'_f$ and $B''_f$,
respectively. We denote by $G'$ the reduced graph.  Let $F$ be a
$k$-completion of $G$. Let $F_1$ and $F_2$ be the completions of $G_1$
and $G_2$ such that $|F_1| + |F_2| \le k$. Moreover, let $H_1 = G_1 +
F_1$ and $H_2 = G_2 + F_2$.  By Lemma~\ref{lem:1branches}, we know
that the vertices of $B' \setminus B'_f$ (resp. $B'' \setminus B''_f$)
can be inserted into the umbrella ordering $\sigma_{H_1}$ of $H_1$
(resp. $\sigma_{H_2}$) in the same order than in $B'$
(resp. $B''$). We thus obtain two proper interval graphs $H'_1$ and
$H'_2$ whose respective umbrella ordering preserve the order of $B'$
and $B''$. We now connect $H'_1$ and $H'_2$ by putting back the edges
contained in $E(B',B'')$, obtaining a graph $H$ with ordering
$\sigma_H$.  Since $G[B]$ is a proper interval graph and $B'$ and
$B''$ are ordered according to $B$ in $H'_1$ and $H'_2$ , it follows
that $H$ is a proper interval graph, and hence $F = F_1 \cup F_2$ is a
$k$-completion of $G$.
\end{proof}
		
\begin{observation}
\label{obs:2branch}
Let $G = (V,E)$ be a positive instance of \PIC{} reduced under
Rules~\ref{rule:twins} to \ref{rule:2branches}. The $2$-branches of
$G$ contain at most $(k + 3)(k^3+4k^2+5k+1)$ vertices.
\end{observation}
		
\begin{proof}
Let $B$ be a $2$-branch of a graph $G = (V,E)$ and $C$ be the connected 
component containing $B$.
If $G[C\setminus B^R]$ is connected, then Observation~\ref{obs:2branch_con}
implies the result. Otherwise, as
$G$ has been reduced under
Rules~\ref{rule:twins} to \ref{rule:2branches}, we know that $|B^R|\le 4k+4$
and then that $|B|\le 2(k^3+4k^2+5k+1)+(4k+4)$ which is less than
$(k + 3)(k^3+4k^2+5k+1)$, provided that $k\ge 1$.
\end{proof}

\subsection{Detecting the branches}

We now turn our attention to the complexity needed to compute
reduction rules~\ref{rule:$K$-join} to~\ref{rule:2branches}.
Mainly, we indicate how to obtain the maximum branches in order to reduce them.
The detection of a branch is straightforward except for the attachment
cliques, where several choices are possible.\\
So, first, we detect the maximum 1-branches of $G$. Remark that
for every vertex $x$ of $G$, the set $\{x\}$ is a 1-branch
of $G$. The next lemma indicates how to compute a maximum
1-branch that contains a fixed vertex $x$ as first vertex.

\begin{lemma}
\label{lem:find1branches}
Let $G = (V,E)$ be a graph and $x$ a vertex of $G$.
In time $O(n^2)$, it is possible to detect a maximum 1-branch of
$G$ containing $x$ as first vertex.
\end{lemma}

\begin{proof}
To detect such a 1-branch, we design an algorithm which has two parts.
Roughly speaking, we first try to detect the set $B^R$ of a 1-branch
$B$ containing $x$. We set $B^R_0=\{x\}$ and $\sigma_0=x$. Once
$B^R_{i-1}$ has been defined, we construct the set $C_i$ of vertices
of $G\setminus (\cup_{l=1}^{i-1} B^R_l )$ that are adjacent to at
least one vertex of $B^R_{i-1}$.  Two cases can appear. First, assume
that $C_i$ is a clique and that it is possible to order the vertices of
$C_i$ such that for every $1 \leqslant j < |C_i|$, we have
$N_{B^R_{i-1}}(c_{j+1}) \subseteq N_{B^R_{i-1}}(c_j)$ and $(N_G(c_j)
\setminus B^R_{i-1}) \subseteq (N_G(c_{j+1}) \setminus B^R_{i-1})$.
In this case, the vertices of $C_i$ correspond to a new $K$-join of
the searched 1-branch (remark that, along this inductive construction,
there is no edge between $C_i$ and $\cup_{l=1}^{i-2}B^R_l$). So, we
let $B^R_i = C_i$ and $\sigma_i$ be the concatenation of
$\sigma_{i-1}$ and the ordering defined on $C_i$.  In the other case,
such an ordering of $C_i$ can not be found, meaning that while
detecting a $1$-branch $B$, we have already detected the vertices of
$B^R$ and at least one (possibly more) vertex of the attachment clique
$B_1$ with neighbors in $B^R$.  Assume that the process stops at step
$p$ and let $C$ be the set of vertices of $G\setminus \cup_{l=1}^{p}
B_l^R$ which have neighbors in $\cup_{l=1}^{p} B_l^R$ and
$B'_1\subseteq B_p^R$ be the set of vertices that are adjacent to all
the vertices of $C$. Remark that $B'_1\neq \emptyset$, as $B'_1$
contains at least the last vertex of $\sigma_p$. We denote by $B^R$
the set $(\cup_{l=1}^p B^R_l) \setminus B'_1$ and we will construct
the largest $K$-join containing $B'_1$ in $G \setminus B^R$ which is
compatible with $\sigma_p$, in order to define the attachment clique
$B_1$ of the desired 1-branch. The vertices of $C$ are the candidates
 to complete the attachment
clique. On $C$, we define the following oriented graph: there is an
arc from $x$ to $y$ if: $xy$ is an edge of $G$,
$N_{B^R}(y)\subseteq N_{B^R}(x)$ and $N_{G\setminus B^R}[x]\subseteq
N_{G\setminus B^R}[y]$. This graph can be computed in time
$O(n^2)$. Now, it is easy to check that the obtained oriented graph is
a transitive graph, in which the equivalent classes are made of true
twins in $G$.  A path in this oriented graph corresponds, by
definition, to a $K$-join containing $B_1'$ and compatible with
$\sigma_p$. As it is possible to compute a longest path in linear time
in this oriented graph, we obtain a maximum 1-branch of $G$ that
contains $x$ as first vertex.
\end{proof}

Now, to detect the 2-branches, we first detect for all pairs of
vertices a maximum $K$-join with these vertices as ends.  More
precisely, if $\{x,y\}$ are two vertices of $G$ linked by an edge,
then $\{x,y\}$ is a $K$-join of $G$, with $N=N_G(x)\cap N_G(y)$,
$L=N_G(x)\setminus N_G[y]$ and $R=N_G(y)\setminus N_G[x]$.  So,
there exist $K$-joins with $x$ and $y$ as ends, and we will compute
such a $K$-join with maximum cardinality.

\begin{lemma}
\label{lem:bijoins}
Let $G = (V, E)$ be a graph and $x$ and $y$ two adjacent vertices of
$G$.  It is possible to compute in cubic time a maximum (in
cardinality) $K$-join that admits $x$ and $y$ as ends.
\end{lemma}

\begin{proof}
We denote $N_G[x]\cap N_G[y]$ by $N$, $N_G(x)\setminus N_G[y]$ by $L$
and $N_G(y)\setminus N_G[x]$ by $R$. Let us denote by $N'$ the set of
vertices of $N$ that contains $N$ in their closed neighborhood.  The
vertices of $N'$ are the candidates to belong to the desired
$K$-join. Now, we construct on $N'$ an oriented graph, putting, for
every vertices $u$ and $v$ of $N'$, an arc from $u$ to $v$ if:
 $N_G(v)\cap L \subseteq N_G(u)\cap L$ and $N_G(u)\cap R \subseteq
N_G(v)\cap R$. Basically, it could take a $O(n)$ time to decide if
there is an arc from $u$ to $v$ or not, and so the whole oriented graph 
could be computed in time $O(n^3)$. Now, it is easy to check that
the obtained oriented graph is a transitive graph in which the
equivalent classes are made of true twins in $G$. In this oriented
graph, it is possible to compute a longest path from $x$ to $y$ in
linear time. Such a path corresponds to a maximal $K$-join that admits
$x$ and $y$ as ends. It follows that the desired $K$-join can be identified in
$O(n^3)$ time.
\end{proof}

Now, for every edge $xy$ of $G$, we compute a maximum $K$-join
that contains $x$ and $y$ as ends and a reference to all the vertices
that this $K$-join contains. This computation takes a $O(n^3m)$ time and
gives, for every vertex, some maximum $K$-joins that contain this vertex.
These $K$-joins will be useful to compute the 2-branches of $G$,
in particular through the next lemma.

\begin{lemma}
\label{lem:maxBijoinIn2branch}
Let $B$ be a 2-branch of $G$ with $B^R\neq \emptyset$, and $x$ a
vertex of $B^R$. Then, for every maximal (by inclusion) $K$-join $B'$
that contains $x$ there exists an extremal edge $uv$ of $\sigma_B$ such that
$B'= \{w \in B\ :\ u\le_{\sigma_B} w \le_{\sigma_B} v\}$.
\end{lemma}

\begin{proof}
As usually, we denote by $L$, $R$ and $C$ the partition of $G\setminus
B$ associated with $B$ and by $\sigma_B$ the umbrella ordering
associated with $B$.  Let $B'$ be a maximal $K$-join that contains $x$
and define by $b_f$ (resp. $b_l$) the first (resp. last) vertex of
$B'$ according to $\sigma_B$. As there is no edge between $\{u\in B
\ :\ u<_{\sigma_B} b_f \}\cup L \cup C$ and $b_l$ and no edge between
$\{u\in B \ :\ b_l<_{\sigma_B} u \}\cup R \cup C$ and $b_f$, we have
$B'\subseteq \{u\in B \ :\ b_f\le_{\sigma_B} u \le
b_l\}$. Furthermore, as $\{u\in B \ :\ b_f\le_{\sigma_B} u \le b_l\}$
is a $K$-join and $B'$ is maximal, we have $B'=\{u\in B
\ :\ b_f\le_{\sigma_B} u \le b_l\}$.  Now, if $b_fb_l$ was not an
extremal edge of $\sigma_B$, it would be possible to extend $B'$,
contradicting the maximality of $B'$.
\end{proof}

Now, we can detect the 2-branches $B$ with a set $B^R$ non
empty. Observe that this is enough for our purpose since we want to
detect $2$-branches of size at least $(k + 3)(k^3+4k^2+5k+1)$
and the attachment cliques contain at most $2(k^3+4k^2+7k+3)$
vertices.

\begin{lemma}
\label{lem:find2branches}
Let $G = (V,E)$ be a graph, $x$ a vertex of $G$ and $B'$ a given
maximal $K$-join that contains $x$. There is a quadratic time algorithm
to decide if there exists a 2-branch $B$ of $G$ which contains $x$ as
a vertex of $B^R$, and if it exists, to find a maximum 2-branch with
this property.
\end{lemma}

\begin{proof}
By Lemma~\ref{lem:maxBijoinIn2branch}, if there exists a 2-branch $B$
of $G$ which contains $x$ as a vertex of $B^R$, then $B'$ corresponds
to a set $\{u\in B \ :\ b_f\le_{\sigma_B} u \le_{\sigma_B} b_l\}$
where $b_fb_l$ is an extremal edge of $B$. We denote by $L'$,
$R'$ and $C'$ the usual partition of $G\setminus B'$ associated with
$B'$, and by $\sigma_{B'}$ the umbrella ordering of $B'$.  In $G$, we
remove the set of vertices $\{u\in B' \ :\ u<_{\sigma_{B'}} x\}$ and
the edges between $L'$ and $\{u\in B' \ :\ x\le_{\sigma_{B'}} u\}$ and
denote by $H_1$ the resulting graph. From the definition of the
2-branch $B$, $\{u\in B \ :\ x\le_{\sigma_{B}} u\}$ is a 1-branch of
$H_1$ that contains $x$ as first vertex. So, using
Lemma~\ref{lem:find1branches}, we find a maximal 1-branch $B_1$ that
contains $x$ as first vertex.  Remark that $B_1$ has to contain
$\{u\in B \ :\ x\le_{\sigma_{B}} u\}\cap B^R$ at its
beginning. Similarly, we define $H_2$ from $G$ by removing the vertex
set $\{u\in B' \ :\ x<_{\sigma_{B'}} u\}$ and the edges between $R'$
and $\{u\in B' \ :\ u\le_{\sigma_{B'}} x\}$. We detect in $H_2$ a
maximum 1-branch $B_2$ that contains $x$ as last vertex, and as
previously, $B_2$ has to contain $\{u\in B \ :\ u\le_{\sigma_{B}}
x\}\cap B^R$ at its end. So, $B_1\cup B_2$ forms a maximum
2-branch of $G$ containing $x$.
\end{proof}

We would like to mention that it could be possible to improve
 the execution time of our detecting branches
algorithm, using possibly more involved techniques (as for instance, inspired
from~\cite{Corn95}). However, this is not our main objective here.\\
Anyway,  using a $O(n^4)$ brute force detection to localize all
the 4-cycles and the claws,  we obtain the following result.

\begin{lemma}
\label{lem:rulespoly}
Given a graph $G = (V,E)$, the reduction rules~\ref{rule:$K$-join}
to~\ref{rule:2branches} can be carried out in polynomial time, namely in 
time $O(n^3m)$.
\end{lemma} 

\subsection{Kernelization algorithm}
	
We are now ready to the state the main result of this Section. The
kernelization algorithm consists of an exhaustive application of
Rules~\ref{rule:cc} to~\ref{rule:2branches}.
		
% the PIC kernel
\begin{theorem}
The \PIC{} problem admits a kernel with $O(k^5)$ vertices.
\end{theorem}
		
\begin{proof}
Let $G = (V,E)$ be a positive instance of \PIC{} reduced under
Rules~\ref{rule:cc} to~\ref{rule:2branches}.  Let $F$ be a
$k$-completion of $G$, $H = G + F$ and $\sigma_H$ be the umbrella
ordering of $H$. Since $|F| \le k$, $G$ contains at most $2k$
\emph{affected} vertices (i.e. incident to an added edge). Let $A =
\{a_1 <_{\sigma_H} \ldots <_{\sigma_H} a_i <_{\sigma_H} \ldots
<_{\sigma_H} a_p\}$ be the set of such vertices, with $p \le 2k$. The
size of the kernel is due to the following observations (see
Figure~\ref{fig:size}):
\begin{itemize}
\item Let $L_0 = \{l \in V\ :\ l <_{\sigma_H} a_1\}$ and $R_{p+1} =
  \{r \in V\ :\ a_p <_{\sigma_H} r\}$.  Since the vertices of $L_0$ and
  $R_{p+1}$ are not affected, it follows that $G[L_0]$ and
  $G[R_{p+1}]$ induce a proper interval graph. As Rule~\ref{rule:cc}
  has been applied, $G[L_0]$ and $G[R_{p+1}]$ both contain one
  connected component,  and $L_0$ and $R_{p+1}$ are $1$-branches of
  $G$. So, by Observation~\ref{obs:1branch}, $L_0$ and $R_{p+1}$ both
  contain at most $k^3+4k^2+9k+4$ vertices.
\item Let $S_i = \{s \in V\ :\ a_i <_{\sigma_H} s <_{\sigma_H}
  a_{i+1}\}$ for every $1 \le i < p$.  Again, since the vertices
  of $S_i$ are not affected, it follows that $G[S_i]$ is a proper
  interval graph. As Rule~\ref{rule:cc} as been applied, there are
at most two connected components in $G[S_i]$. If $G[S_i]$ is connected, 
then, $S_i$ is a $2$-branch of $G$ and, by Observation~\ref{obs:2branch}, $S_i$
  contains at most $(k+3)(k^3+4k^2+5k+1)$ vertices.
Otherwise, if $G[S_i]$ contains two connected components, they correspond
to two 1-branches of $G$, and by Observation~\ref{obs:1branch}, $S_i$ 
contain at most $2(k^3+4k^2+9k+4)$ vertices. In both cases, we bound the number 
of vertices of $S_i$ by  $(k+3)(k^3+4k^2+5k+1)$, provided that $k\ge 1$.
\end{itemize}
			
\begin{figure}[ht]
\centerline{
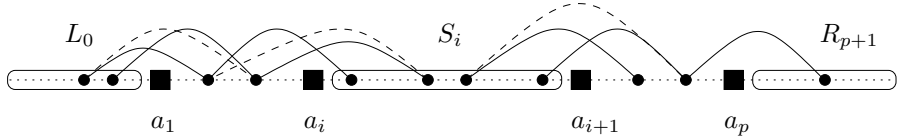}
\caption{Illustration of the size of the kernel. The figure represents
  the graph $H=G+F$, the square vertices stand for the \emph{affected}
  vertices, $L_0$ and $R_{p+1}$ are $1$-branches of $G$, and, on the
  figure, $S_i$ defines a $2$-branch. \label{fig:size}}
\end{figure}
			
Altogether, the proper interval graph $H$ (and hence $G$) contains at
most:
			
$$ 2( k^3+4k^2+9k+4 ) + (2k - 1)( (k + 3)(k^3+4k^2+5k+1) )$$
			
vertices, which implies the claimed $O(k^5)$ bound. The complexity
directly follows from Lemma~\ref{lem:rulespoly}.
		
\end{proof}

% ======================================================================================================================== %	

\section{A special case: \sc{\BCC}}

\emph{Bipartite chain graphs} are defined as bipartite graphs whose
parts are connected by a join. Equivalently, they are known to be the
graphs that do not admit any $\{2K_2,C_5,K_3\}$ as an induced
subgraph~\cite{Yan81} (see Figure~\ref{fig:forbiddenbcg}).
In~\cite{Guo07}, Guo proved that the so-called \textsc{Bipartite Chain
 Deletion With Fixed Bipartition} problem, where one is given a
\emph{bipartite} graph $G = (V,E)$ and seeks a subset of $E$ of size
at most $k$ whose deletion from $E$ leads to a bipartite chain graph,
admits a kernel with $O(k^2)$ vertices. We define \emph{bi-clique chain 
graph} to be the graphs formed by two disjoint cliques linked by a join.
They correspond to interval graphs that can be covered by two cliques.
Since the complement
of a bipartite chain graph is a bi-clique chain graph, this
result also holds for the \textsc{Bi-clique Chain Completion With
  Fixed Bi-clique Partition} problem.  Using similar techniques than in
Section~\ref{sec:pic}, we prove that when the bipartition is not
fixed, both problems admit a quadratic-vertex kernel.  For the sake of
simplicity, we consider the completion version of the problem,
defined as follows.\\
\\ 
\BCC{}: \\ 
\textbf{Input}: A graph $G = (V,E)$ and a positive integer $k$. \\ 
\textbf{Parameter}: $k$. \\ 
\textbf{Output}: A set $F \subseteq (V \times V) \setminus E$
of size at most $k$ such that the graph $H = (V, E \cup F)$ is a
bi-clique chain graph. \\
	
It follows from definition that bi-clique chain graphs do not admit
any $\{C_4, C_5, 3K_1\}$ as an induced subgraph, where a $3K_1$ is an
independent set of size $3$ (see
Figure~\ref{fig:forbiddenbcg}). Observe in particular that bi-clique
chain graphs are proper interval graphs, and hence admit an umbrella
ordering.
	
\begin{figure}[ht]
\centerline{
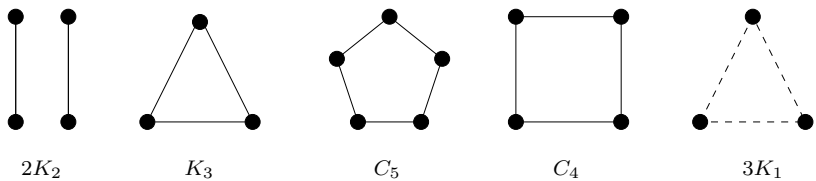}
\caption{The forbidden induced subgraphs for bipartite and bi-clique
  chain graphs.\label{fig:forbiddenbcg}}
\end{figure}
	 
We provide a kernelization algorithm for the \BCC{} problem which
follows the same lines that the one in Section~\ref{sec:pic}.
	
\begin{polyrule}[Sunflower]
\label{rule:sunflowerbcc}
Let $\mathcal{S} = \{C_1, \ldots, C_m\}$, $m > k$ be a set of $3K_1$
having two vertices $u,v$ in common but distinct third vertex. Add
$uv$ to $F$ and decrease $k$ by $1$.\\ 
Let $\mathcal{S} = \{C_1, \ldots, C_m\}$, $m > k$ be a set of distinct
$4$-cycles having a non-edge $uv$ in common. Add $uv$ to $F$ and
decrease $k$ by $1$.
\end{polyrule}
	
The following result is similar to
Lemma~\ref{lem:claws_and_$4$-cycles}.
	
% prouver la regle du $K$-join pour voir si on a besoin des $4$-cycles.
\begin{lemma}
\label{lem:s3}
Let $G = (V,E)$ be a positive instance of \BCC{} on which
Rule~\ref{rule:sunflowerbcc} has been applied. There are at most
$k^2+2k$ vertices of $G$ contained in $3K_1$'s. Furthermore, there at most 
$2k^2+2k$ vertices of $G$ that are vertices of a $4$-cycle.
\end{lemma}
	
We say that a $K$-join is \emph{simple} whenever $L = \emptyset$ or $R
= \emptyset$. In other words, a simple $K$-join consists in a clique
connected to the rest of the graph by a join.  We will see it as a
1-branch which is a clique and use for it the classical notation
devoted to the 1-branch. Moreover, we (re)define a \emph{clean
  $K$-join} as a $K$-join whose vertices do not belong to any $3K_1$
or $4$-cycle.  The following reduction rule is similar to
Rule~\ref{rule:$K$-join}, the main ideas are identical, only some
technical arguments change. Anyway, to be clear, we give the proof in
all details.

\begin{polyrule}[$K$-join]
\label{rule:simple$K$-join}
Let $B$ be a simple clean $K$-join of size at least $2(k+1$)
associated with an umbrella ordering $\sigma_B$. Let $B_L$
(resp. $B_R$) be the $k + 1$ first (resp. last) vertices of $B$
according to $\sigma_B$, and $M = B \setminus (B_L \cup B_R)$. Remove
the set of vertices $M$ from $G$.
\end{polyrule}
		
% $K$-join rule is safe
\begin{lemma}
\label{lem:$K$-joinbcc}
Rule~\ref{rule:simple$K$-join} is safe and can be computed in polynomial time.
\end{lemma}

\begin{proof}
Let $G' = G \setminus M$. Observe that any $k$-completion of $G$ is a
$k$-completion of $G'$ since bi-clique chain graphs are closed under
induced subgraphs. So, let $F$ be a $k$-completion for $G'$. We denote
by $H = G' + F$ the resulting bi-clique chain graph and by $\sigma_H$
an umbrella ordering of $H$. We prove that we can always insert the
vertices of $M$ into $\sigma_H$ and modify it if necessary, to obtain
an umbrella ordering of a bi-clique chain graph for $G$ without adding
any edge. This will imply that $F$ is a $k$-completion for $G$.  To
see this, we need the following structural property of $G$. As usual,
we denote by $R$ the neighbors in $G\setminus B$ of the vertices of
$B$, and by $C$ the vertices of $G\setminus (R\cup B)$. For the sake
of simplicity, we let $N = \cap_{b \in B} N_G(b) \setminus B$, and
remove the vertices of $N$ from $R$. We abusively still denote by $R$
the set $R\setminus N$, see Figure~\ref{fig:Kjoin-biclique}.

\begin{figure}[ht]
\centerline{
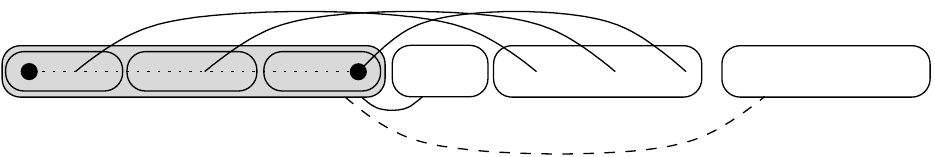}
\caption{The $K$-join decomposition for the \BCC{}
  problem. \label{fig:Kjoin-biclique}}
\end{figure}

\begin{claim}
\label{claim:structbcc}
The set $R \cup C$ is a clique of $G$.
%% in $G$.
\end{claim}
			
\emph{Proof.}  Observe that no vertex of $R$ is a neighbor of $b_1$,
since otherwise such a vertex must be adjacent to all the vertices of
$B$ and then must stand in $N$.  So, if $R\cup C$ contains two
vertices $u,v$ such that $uv \notin E$, we form the $3K_1$
$\{b_1,u,v\}$, contradicting the fact that $B$ is clean.  \hfill
$\diamond$ \\
			
The following observation comes from the definition of a simple $K$-join.
			
\begin{observation}
\label{obs:$K$-joinbcc}
Given any vertex $r \in R$, if $N_B(r) \cap B_L \neq \emptyset$ holds
then $M \subseteq N_B(r)$. 
\end{observation}

We use these facts to prove that an umbrella ordering of a bi-clique
chain graph can be obtained for $G$ by inserting the vertices of $M$
into $\sigma_H$. Let $b_f, b_l$ be the first and last vertex of $B
\setminus M$ appearing in $\sigma_H$, respectively. We let $B_H$
denote the set $\{u \in V(H)\ :\ b_f <_{\sigma_H} u <_{\sigma_H}
b_l\}$.  Now, we modify $\sigma_H$ by ordering the twins in $H$
according to their neighborhood in $M$: if $x$ and $y$ are twins in
$H$, are consecutive in $\sigma_H$, verify
$x<_{\sigma_H}y<_{\sigma_H}b_f$ and $N_M(y)\subset N_M(x)$, then we
exchange $x$ and $y$ in $\sigma_H$. This process stops when the
considered twins are ordered following the join between $\{u \in
V(H)\ :\ u <_{\sigma_H} b_f\}$ and $M$.  We proceed similarly on the
right of $B_H$, i.e. for $x$ and $y$ consecutive twins with
$b_l<_{\sigma_H}x<_{\sigma_H}y$ and $N_M(x)\subset N_M(y)$.  The
obtained order is clearly an umbrella ordering of a bi-clique chain
graph too (in fact, we just re-labeled some vertices in $\sigma_H$,
and we abusively still denote it by $\sigma_H$).\\
			
\begin{claim}
\label{claim:cliquebcc}
The set $B_H \cup \{m\}$ is a clique of $G$ for any $m \in M$, and 
consequently $B_H\cup M$ is a clique of $G$.
\end{claim}
			
\emph{Proof.}  Let $u$ be any vertex of $B_H$. We claim that $um \in
E(G)$.  Observe that if $u \in B$ then the claim trivially holds. So,
assume that $u \notin B$. By definition of $\sigma_H$, $B_H$ is a clique in
$H$ since $b_fb_l \in E(G)$. It follows that $u$ is incident to every
vertex of $B\setminus H$ in $H$. Since $B_L$ contains $k + 1$ vertices, it
follows that $N_G(u) \cap B_L \neq \emptyset$. Hence, $u$ belongs to
$N \cup R$ and $um \in E$ by Observation~\ref{obs:$K$-join}.  \hfill
$\diamond$
			
\begin{claim}
\label{claim:umbrellambcc}
Let $m$ be any vertex of $M$ and $\sigma'_H$ be the ordering obtained
from $\sigma_H$ by removing $B_H$ and inserting $m$ to the position of
$B_H$. The ordering $\sigma'_H$ respects the umbrella property.
\end{claim}
			
\emph{Proof.}  Assume that $\sigma'_H$ does not respect the umbrella
property, i.e. that there exist (w.l.o.g.) two vertices $u, v\in
H\setminus B_H$ such that either $(1)$ $u<_{\sigma'_H} v <_{\sigma'_H}
m$, $um \in E(H)$ and $uv\notin E(H)$ or $(2)$ $u<_{\sigma'_H} m
<_{\sigma'_H} v$, $um\notin E(H)$ and $uv\in E(H)$ or $(3)$
$u<_{\sigma'_H} v <_{\sigma'_H} m$, $um \in E(H)$ and $vm \notin
E(H)$. First, assume that $(1)$ holds. Since $uv \notin E$ and
$\sigma_H$ is an umbrella ordering, $uw \notin E(H)$ for any $w \in
B_H$, and hence $uw \notin E(G)$.  This means that $B_R \cap N_G(u) =
\emptyset$, which is impossible since $um \in E(G)$.  If $(2)$ holds,
since $uv\in E(H)$ and $\sigma_H$ is an umbrella ordering of $H$, we
have $B_H\subseteq N_H(u)$. In particular, $B_L\subseteq N_H(u)$
holds, and as $|B_L|=k+1$, we have $B_L\cap N_G(u)\neq \emptyset$
and $um$ should be an edge of $G$, what contradicts the assumption
$um\notin E(H)$.  So, $(3)$ holds, and we choose the first $u$
satisfying this property according to the order given by
$\sigma'_H$. So we have $wm \notin E(G)$ for any $w <_{\sigma'_H}
u$. Similarly, we choose $v$ to be the first vertex satisfying $vm
\notin E(G)$. Since $um \in E(G)$, we know that $u$ belongs to $N \cup
R$.  Moreover, since $vm \notin E(G)$, $v \in R\cup C$.  There are
several cases to consider:
				
% TODO: find a better way to denote the neighborhoods
\begin{enumerate}[(i)]
\item $u \in N$: in this case we know that $B \subseteq N_G(u)$, and
  in particular that $ub_l \in E(G)$.  Since $\sigma_H$ is an umbrella
  ordering for $H$, it follows that $vb_l \in E(H)$ and that
  $B_L\subseteq N_H(v)$. Since $|B_L| = k + 1$ we know that $N_G(v) \cap
  B_L \neq \emptyset$ and hence $v \in R$.  It follows from
  Observation~\ref{obs:$K$-join} that $vm \in E(G)$.
\item $u \in R, v \in R\cup C$: in this case $uv\in E(G)$, by
  Claim~\ref{claim:structbcc}, but $u$ and $v$ are not true twins in
  $H$ (otherwise $v$ would be placed before $u$ in $\sigma_H$ due to
  the modification we have applied to $\sigma_H$). This means that
  there exists a vertex $w \in V(H)$ that \emph{distinguishes} $u$
  from $v$ in $H$.
  
  Assume first that $w <_{\sigma_H} u$ and that $uw \in E(H)$ and $vw
  \notin E(H)$. We choose the first $w$ satisfying this according to
  the order given by $\sigma'_H$. Since $vm, wm, vw \notin E(H)$, it
  follows that $\{v,w,m\}$ defines a $3K_1$ of $G$, which cannot be
  since $B$ is clean.  Hence we can assume that for any $w''
  <_{\sigma_H} u$, $uw'' \in E(H)$ implies that $vw'' \in E(H)$. Now,
  suppose that $b_l <_{\sigma_H} w$ and $uw \notin E(H),\ vw \in
  E(H)$. In particular, this means that $B_L \subseteq N_H(v)$.  Since
  $|B_L| = k + 1$ we have $N_G(v) \cap B_L \neq \emptyset$, implying
  $vm \in E(G)$ (Observation~\ref{obs:$K$-join}). Assume now that $v
  <_{\sigma_H} w <_{\sigma_H} b_f$.  In this case, since $uw \notin
  E(H)$, $B \cap N_H(u) = \emptyset$ holds and hence $B \cap N_G(u) =
  \emptyset$, which cannot be since $u \in R$. Finally, assume that $w
  \in B_H$ and choose the last vertex $w$ satisfying this according to
  the order given by $\sigma'_H$ (i.e. $vw' \notin E(H)$ for any $w
  <_{\sigma_H} w'$ and $w'\in B_H$).  If $vw \in E(G)$ then
  $\{u,m,w,v\}$ is a $4$-cycle in $G$ containing a vertex of $B$,
  which cannot be (recall that $B_H \cup \{m\}$ is a clique of $G$ by
  Claim~\ref{claim:clique}).  Hence $vw \in F$ and there exists an
  extremal edge above $vw$. The only possibility is that this edge is
  some edge $u'w$ for some $u'$ with $u' \in V(H)$, $u <_{\sigma_H} u'
  <_{\sigma_H} v$ and $u'w \in E(G)$. By the choice of $v$ we know
  that $u'm \in E(G)$. Moreover, by the choice of $w$, observe that
  $u'$ and $v$ are true twins in $H$ (if a vertex $s$ distinguishes
  $u'$ and $v$ in $H$, $s$ cannot be before $u$, since otherwise $s$
  would distinguish $u$ and $v$, and not before $w$, by choice of
  $w$). This leads to a contradiction because $v$ should have been
  placed before $u$ through the modification we have applied to
  $\sigma_H$.  \hfill $\diamond$
\end{enumerate}
			
\begin{claim}
\label{claim:propermbcc}
Every vertex $m \in M$ can be added to the graph $H$ while preserving
an umbrella ordering.
\end{claim}

\emph{Proof.}  Let $m$ be any vertex of $M$. The graph $H$ is a
bi-clique chain graph. So, we know that in its associated umbrella
ordering $\sigma_H=b_1,\dots ,b_{|H|}$, there exists a vertex $b_i$
such that $H_1=\{b_1,\dots ,b_i\}$ and $H_2=\{b_{i+1},\dots
,b_{|H|}\}$ are two cliques of $H$ linked by a join.  We study the
behavior of $B_H$ according to the partition $(H_1,H_2)$.
\begin{enumerate}[(i)]
\item \label{item:1properbcc} Assume first that $B_H \subseteq H_1$
  (the case $B_H \subseteq H_2$ is similar). We claim that the set
  $H_1 \cup \{m\}$ is a clique. Indeed, let $v \in H_1 \setminus B_H$:
  since $H_1$ is a clique, $B_H \subseteq N_H(v)$ and hence $N_G(v)
  \cap B_L \neq \emptyset$. In particular, this means that $vm \in
  E(G)$ by Observation~\ref{obs:$K$-joinbcc}.  Since $B_H \cup \{m\}$
  is a clique by Claim~\ref{claim:cliquebcc}, the result follows.
  Now, let $u$ be the neighbor of $m$ with maximal index in
  $\sigma_H$, and $b_{u}$ the neighbor of $u$ with minimal index in
  $\sigma_H$. Observe that we may assume $u \in H_2$ since otherwise
  $N_H(m) \cap H_2 = \emptyset$ and hence we insert $m$ at the
  beginning of $\sigma_H$.  First, if $b_u \in H_1$, we prove that the
  order $\sigma_m$ obtained by inserting $m$ directly before $b_{u}$
  in $\sigma_H$ yields an umbrella ordering of a bi-clique chain
  graph. Since $H_1 \cup \{m\}$ is a clique, we only need to show that
  $N_{H_2}(v) \subseteq N_{H_2}(m)$ for any $v \le_{\sigma_m} b_u$ and
  $N_{H_2}(m) \subseteq N_{H_2}(w)$ for any $w\in H_2$ with $w
  \ge_{\sigma_m} b_u$.  Observe that by Claim~\ref{claim:umbrellambcc}
  the set $\{w \in V\ :\ m \le_{\sigma_m} w \le_{\sigma_m} u\}$ is a
  clique. Hence the former case holds since $vu' \notin E(G)$ for any
  $v \le_{\sigma_m} b_u$ and $u' \ge_{\sigma_m} u$. The latter case
  also holds since $N_H(m) \subseteq N_H(b_u)$ by construction.
  Finally, if $b_u \in H_2$, then $b_u = b_{|H_1|+1}$ since $H_2$ is a
  clique. Hence, using similar arguments one can see that inserting
  $m$ directly after $b_{|H_1|}$ in $\sigma_H$ yields an umbrella
  ordering of a bi-clique chain graph.
\item Assume now that $B_H \cap H_1 \neq \emptyset$ and $B_H \cap H_2
  \neq \emptyset$. In this case, we claim that $H_1 \cup \{m\}$ or
  $H_2 \cup \{m\}$ is a clique in $H$. Let $u$ and $u'$ be the
  neighbors of $m$ with minimal and maximal index in $\sigma_H$,
  respectively. If $u = b_1$ or $u' = b_{|H|}$ then
  Claims~\ref{claim:cliquebcc} and \ref{claim:umbrellambcc} imply that
  $H_1 \cup \{m\}$ or $H_2 \cup \{m\}$ is a clique and we are
  done. So, none of these two conditions hold and $mb_1\notin E(H)$
  and $mb_{|H|}\notin E(H)$ Then, by Claim~\ref{claim:umbrellambcc}, we know 
  that $b_1b_{|H|}$ and the set $\{b_1,b_{|H|},m\}$ defines a $3K_1$
  containing $m$ in $G$, which cannot be.  This means that we can
  assume w.l.o.g. that $H_1 \cup \{m\}$ is a clique, and we can
  conclude using similar arguments than in (\ref{item:1properbcc}).
\end{enumerate}
\hfill $\diamond$ \\
	
Since the proof of Claim~\ref{claim:propermbcc} does not use the fact
that the vertices of $H$ do not belong to $M$, it follows that we can
iteratively insert the vertices of $M$ into $\sigma_H$, preserving an
umbrella ordering at each step. To conclude, observe that the
reduction rule can be computed in polynomial time using
Lemma~\ref{lem:bijoins}.
\end{proof}
	
\begin{observation}
\label{obs:bcc}
Let $G = (V,E)$ be a positive instance of \BCC{} reduced under
Rule~\ref{rule:simple$K$-join}. Any simple $K$-join $B$ of $G$ has
size at most $3k^2+6k+2$.
\end{observation}
	
\begin{proof}
Let $B$ be any simple $K$-join of $G$, and assume $|B| >
3k^2+6k+2$. By Lemma~\ref{lem:s3} we know that at most
$3k^2+2k$ vertices of $B$ are contained in a $3K_1$ or a
$4$-cycle. Hence $B$ contains a set $B'$ of at least $2k+3$ vertices
not contained in any $3K_1$ or a $4$-cycle. Now, since any subset of a
$K$-join is a $K$-join, it follows that $B'$ is a \emph{clean} simple
$K$-join. Since $G$ is reduced under rule~\ref{rule:simple$K$-join},
we know that $|B'| \le 2(k+1)$ what gives a contradiction.
\end{proof}
	
Finally, we can prove that Rules~\ref{rule:sunflowerbcc}
and~\ref{rule:simple$K$-join} form a kernelization algorithm.

\begin{theorem}
\label{thm:bcc}
The \BCC{} problem admits a kernel with $O(k^2)$ vertices.
\end{theorem}
	
\begin{proof}
Let $G = (V,E)$ be a positive instance of \BCC{} reduced under
Rules~\ref{rule:sunflowerbcc} and~\ref{rule:simple$K$-join}, and $F$
be a $k$-completion for $G$.  We let $H = G+F$ and $H_1$, $H_2$ be the
two cliques of $H$. Observe in particular that $H_1$ and $H_2$ both
define simple $K$-joins. Let $A$ be the set of affected vertices of
$G$. Since $|F| \le k$, observe that $|A| \le 2k$.  Let $A_1 = A \cap
H_1$, $A_2 = A \cap H_2$, $A'_1 = H_1 \setminus A_1$ and $A'_2 = H_2
\setminus A_2$ (see Figure~\ref{fig:sizebcc}). Observe that since
$H_1$ is a simple $K$-join in $H$, $A'_1 \subseteq H_1$ is a simple
$K$-join of $G$ (recall that the vertices of $A'_1$ are not
affected). By Observation~\ref{obs:bcc}, it follows that $|A'_1| \le
3k^2+6k+2$. The same holds for $A'_2$ and $H$ contains at most
$2(3k^2+6k+2)+2k$ vertices.
		
\begin{figure}[ht]
\centerline{
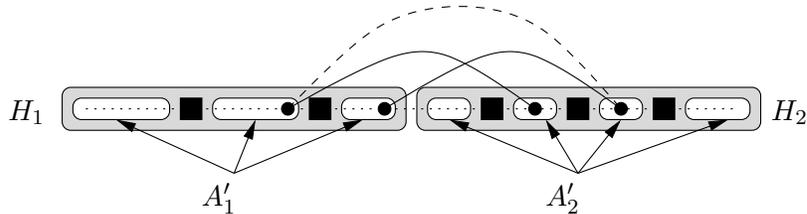}
\caption{Illustration of the bi-clique chain graph $H$. The square
  vertices stand for affected vertices, and the sets $A'_1 = H_1
  \setminus A_1$ and $ A'_2 = H_2 \setminus A_2$ are simple $K$-joins
  of $G$, respectively. \label{fig:sizebcc}}
\end{figure}
\end{proof}
	
\begin{corollary}
\label{coro:bcd}
The \textsc{Bipartite Chain Deletion} problem admits a kernel with
$O(k^2)$ vertices.
\end{corollary}

% ======================================================================================================================== %

\section{Conclusion}

In this paper we prove that the \PIC{} problem admits a kernel with
$O(k^ 5)$ vertices.  Two natural questions arise from our results:
firstly, does the \textsc{Interval Completion} problem admit a
polynomial kernel? Observe that this problem is known to be FPT not
for long~\cite{VHPT09}.  The techniques we developed here intensively
use the fact that there are few claws in the graph, what help us to
reconstruct parts of the umbrella ordering. Of course, these
considerations no more hold in general interval graphs. The second
question is: does the \textsc{Proper Interval Edge-Deletion} problem
admit a polynomial kernel?  Again, this problem admits a
fixed-parameter algorithm~\cite{Vil10a}, and we believe that our
techniques could be applied to this problem as well.  Finally, we
proved that the \BCC{} problem admits a kernel with $O(k^2)$ vertices,
which completes a result of Guo~\cite{Guo07}. In all
cases, a natural question is thus whether these bounds can be
improved?

\bibliographystyle{plain}
\bibliography{biblio}
\end{document}

%% file: umbrella.pdf_t
\begin{picture}(0,0)%
\includegraphics{umbrella.pdf}%
\end{picture}%
\setlength{\unitlength}{4144sp}%
\begingroup\makeatletter\ifx\SetFigFont\undefined%
\gdef\SetFigFont#1#2#3#4#5{%
  \reset@font\fontsize{#1}{#2pt}%
  \fontfamily{#3}\fontseries{#4}\fontshape{#5}%
  \selectfont}%
\fi\endgroup%
\begin{picture}(4212,1218)(664,-409)
\put(1441,-331){\makebox(0,0)[lb]{\smash{{\SetFigFont{12}{14.4}{\rmdefault}{\mddefault}{\updefault}$v_i$}}}}
\put(2566,-331){\makebox(0,0)[lb]{\smash{{\SetFigFont{12}{14.4}{\rmdefault}{\mddefault}{\updefault}$v_l$}}}}
\put(4861,-106){\makebox(0,0)[lb]{\smash{{\SetFigFont{12}{14.4}{\rmdefault}{\mddefault}{\updefault}$\sigma$}}}}
\put(3691,-331){\makebox(0,0)[lb]{\smash{{\SetFigFont{12}{14.4}{\rmdefault}{\mddefault}{\updefault}$v_j$}}}}
\end{picture}%

%% file: 1branch.pdf_t
\begin{picture}(0,0)%
\includegraphics{1branch.pdf}%
\end{picture}%
\setlength{\unitlength}{2763sp}%
\begingroup\makeatletter\ifx\SetFigFont\undefined%
\gdef\SetFigFont#1#2#3#4#5{%
  \reset@font\fontsize{#1}{#2pt}%
  \fontfamily{#3}\fontseries{#4}\fontshape{#5}%
  \selectfont}%
\fi\endgroup%
\begin{picture}(7544,1464)(-321,-1541)
\put(4201,-1261){\makebox(0,0)[lb]{\smash{{\SetFigFont{12}{14.4}{\rmdefault}{\mddefault}{\updefault}$R$}}}}
\put(6151,-1261){\makebox(0,0)[lb]{\smash{{\SetFigFont{12}{14.4}{\rmdefault}{\mddefault}{\updefault}$C$}}}}
\put(2476,-736){\makebox(0,0)[lb]{\smash{{\SetFigFont{8}{9.6}{\rmdefault}{\mddefault}{\updefault}$B_1$}}}}
\put(2176,-1036){\makebox(0,0)[lb]{\smash{{\SetFigFont{8}{9.6}{\rmdefault}{\mddefault}{\updefault}$b_l$}}}}
\put(3151,-1036){\makebox(0,0)[lb]{\smash{{\SetFigFont{8}{9.6}{\rmdefault}{\mddefault}{\updefault}$b_{|B|}$}}}}
\put(-149,-1036){\makebox(0,0)[lb]{\smash{{\SetFigFont{8}{9.6}{\rmdefault}{\mddefault}{\updefault}$b_1$}}}}
\put(1201,-1261){\makebox(0,0)[lb]{\smash{{\SetFigFont{12}{14.4}{\rmdefault}{\mddefault}{\updefault}$B$}}}}
\put(901,-736){\makebox(0,0)[lb]{\smash{{\SetFigFont{8}{9.6}{\rmdefault}{\mddefault}{\updefault}$B^R$}}}}
\end{picture}%

%% file: 2branch.pdf_t
\begin{picture}(0,0)%
\includegraphics{2branch.pdf}%
\end{picture}%
\setlength{\unitlength}{2763sp}%
\begingroup\makeatletter\ifx\SetFigFont\undefined%
\gdef\SetFigFont#1#2#3#4#5{%
  \reset@font\fontsize{#1}{#2pt}%
  \fontfamily{#3}\fontseries{#4}\fontshape{#5}%
  \selectfont}%
\fi\endgroup%
\begin{picture}(7769,1536)(954,-1540)
\put(3676,-1111){\makebox(0,0)[lb]{\smash{{\SetFigFont{8}{9.6}{\rmdefault}{\mddefault}{\updefault}$b_{l'}$}}}}
\put(2926,-736){\makebox(0,0)[lb]{\smash{{\SetFigFont{8}{9.6}{\rmdefault}{\mddefault}{\updefault}$B_1$}}}}
\put(4051,-736){\makebox(0,0)[lb]{\smash{{\SetFigFont{8}{9.6}{\rmdefault}{\mddefault}{\updefault}$B^R$}}}}
\put(5101,-736){\makebox(0,0)[lb]{\smash{{\SetFigFont{8}{9.6}{\rmdefault}{\mddefault}{\updefault}$B_2$}}}}
\put(1501,-1261){\makebox(0,0)[lb]{\smash{{\SetFigFont{12}{14.4}{\rmdefault}{\mddefault}{\updefault}$L$}}}}
\put(2626,-1111){\makebox(0,0)[lb]{\smash{{\SetFigFont{8}{9.6}{\rmdefault}{\mddefault}{\updefault}$b_1$}}}}
\put(4726,-1111){\makebox(0,0)[lb]{\smash{{\SetFigFont{8}{9.6}{\rmdefault}{\mddefault}{\updefault}$b_l$}}}}
\put(5851,-1111){\makebox(0,0)[lb]{\smash{{\SetFigFont{8}{9.6}{\rmdefault}{\mddefault}{\updefault}$b_{|B|}$}}}}
\put(6751,-1261){\makebox(0,0)[lb]{\smash{{\SetFigFont{12}{14.4}{\rmdefault}{\mddefault}{\updefault}$R$}}}}
\put(7951,-1261){\makebox(0,0)[lb]{\smash{{\SetFigFont{12}{14.4}{\rmdefault}{\mddefault}{\updefault}$C$}}}}
\put(4051,-211){\makebox(0,0)[lb]{\smash{{\SetFigFont{12}{14.4}{\rmdefault}{\mddefault}{\updefault}$B$}}}}
\end{picture}%

%% file: bijoin.pdf_t
\begin{picture}(0,0)%
\includegraphics{bijoin.pdf}%
\end{picture}%
\setlength{\unitlength}{3937sp}%
\begingroup\makeatletter\ifx\SetFigFont\undefined%
\gdef\SetFigFont#1#2#3#4#5{%
  \reset@font\fontsize{#1}{#2pt}%
  \fontfamily{#3}\fontseries{#4}\fontshape{#5}%
  \selectfont}%
\fi\endgroup%
\begin{picture}(3292,1854)(155,-1298)
\put(1621, 98){\makebox(0,0)[lb]{\smash{{\SetFigFont{11}{13.2}{\rmdefault}{\mddefault}{\updefault}$M$}}}}
\put(408,111){\makebox(0,0)[lb]{\smash{{\SetFigFont{9}{10.8}{\rmdefault}{\mddefault}{\updefault}$b_1$}}}}
\put(2888,111){\makebox(0,0)[lb]{\smash{{\SetFigFont{9}{10.8}{\rmdefault}{\mddefault}{\updefault}$b_{|B|}$}}}}
\put(1581,397){\makebox(0,0)[lb]{\smash{{\SetFigFont{11}{13.2}{\rmdefault}{\mddefault}{\updefault}$B$}}}}
\put(3008,-1229){\makebox(0,0)[lb]{\smash{{\SetFigFont{11}{13.2}{\rmdefault}{\mddefault}{\updefault}$R$}}}}
\put(1662,-910){\makebox(0,0)[lb]{\smash{{\SetFigFont{11}{13.2}{\rmdefault}{\mddefault}{\updefault}$N$}}}}
\put(315,-1229){\makebox(0,0)[lb]{\smash{{\SetFigFont{11}{13.2}{\rmdefault}{\mddefault}{\updefault}$L$}}}}
\put(2186, 97){\makebox(0,0)[lb]{\smash{{\SetFigFont{11}{13.2}{\rmdefault}{\mddefault}{\updefault}$B_R$}}}}
\put(946,100){\makebox(0,0)[lb]{\smash{{\SetFigFont{11}{13.2}{\rmdefault}{\mddefault}{\updefault}$B_L$}}}}
\end{picture}%

%% file: umbrellam.pdf_t
\begin{picture}(0,0)%
\includegraphics{umbrellam.pdf}%
\end{picture}%
\setlength{\unitlength}{3315sp}%
\begingroup\makeatletter\ifx\SetFigFont\undefined%
\gdef\SetFigFont#1#2#3#4#5{%
  \reset@font\fontsize{#1}{#2pt}%
  \fontfamily{#3}\fontseries{#4}\fontshape{#5}%
  \selectfont}%
\fi\endgroup%
\begin{picture}(8013,2326)(910,-1700)
\put(2296,-1636){\makebox(0,0)[lb]{\smash{{\SetFigFont{10}{12.0}{\rmdefault}{\mddefault}{\updefault}(a)}}}}
\put(6346,479){\makebox(0,0)[lb]{\smash{{\SetFigFont{10}{12.0}{\rmdefault}{\mddefault}{\updefault}$m$}}}}
\put(5041,-1276){\makebox(0,0)[lb]{\smash{{\SetFigFont{10}{12.0}{\rmdefault}{\mddefault}{\updefault}$u$}}}}
\put(6391,-1276){\makebox(0,0)[lb]{\smash{{\SetFigFont{10}{12.0}{\rmdefault}{\mddefault}{\updefault}$v$}}}}
\put(5716,-1276){\makebox(0,0)[lb]{\smash{{\SetFigFont{10}{12.0}{\rmdefault}{\mddefault}{\updefault}$u'$}}}}
\put(6436,-1636){\makebox(0,0)[lb]{\smash{{\SetFigFont{10}{12.0}{\rmdefault}{\mddefault}{\updefault}(b)}}}}
\put(8461,-1276){\makebox(0,0)[lb]{\smash{{\SetFigFont{10}{12.0}{\rmdefault}{\mddefault}{\updefault}$w'$}}}}
\put(7651,-1276){\makebox(0,0)[lb]{\smash{{\SetFigFont{10}{12.0}{\rmdefault}{\mddefault}{\updefault}$w \in B_H$}}}}
\put(2251,479){\makebox(0,0)[lb]{\smash{{\SetFigFont{10}{12.0}{\rmdefault}{\mddefault}{\updefault}$m$}}}}
\put(2791,-1276){\makebox(0,0)[lb]{\smash{{\SetFigFont{10}{12.0}{\rmdefault}{\mddefault}{\updefault}$u'$}}}}
\put(2341,-1276){\makebox(0,0)[lb]{\smash{{\SetFigFont{10}{12.0}{\rmdefault}{\mddefault}{\updefault}$u$}}}}
\put(3286,-1276){\makebox(0,0)[lb]{\smash{{\SetFigFont{10}{12.0}{\rmdefault}{\mddefault}{\updefault}$v$}}}}
\put(1801,-1276){\makebox(0,0)[lb]{\smash{{\SetFigFont{10}{12.0}{\rmdefault}{\mddefault}{\updefault}$w$}}}}
\end{picture}%

%% file: x1x2.pdf_t
\begin{picture}(0,0)%
\includegraphics{x1x2.pdf}%
\end{picture}%
\setlength{\unitlength}{3108sp}%
\begingroup\makeatletter\ifx\SetFigFont\undefined%
\gdef\SetFigFont#1#2#3#4#5{%
  \reset@font\fontsize{#1}{#2pt}%
  \fontfamily{#3}\fontseries{#4}\fontshape{#5}%
  \selectfont}%
\fi\endgroup%
\begin{picture}(9204,2794)(79,-1673)
\put(8761,-1141){\makebox(0,0)[lb]{\smash{{\SetFigFont{9}{10.8}{\rmdefault}{\mddefault}{\updefault}$v_{j+1}$}}}}
\put(8581,-1141){\makebox(0,0)[lb]{\smash{{\SetFigFont{9}{10.8}{\rmdefault}{\mddefault}{\updefault}$v_{j}$}}}}
\put(2161,974){\makebox(0,0)[lb]{\smash{{\SetFigFont{9}{10.8}{\rmdefault}{\mddefault}{\updefault}$m$}}}}
\put(7111,-376){\makebox(0,0)[lb]{\smash{{\SetFigFont{9}{10.8}{\rmdefault}{\mddefault}{\updefault}$m$}}}}
\put(3736,-1141){\makebox(0,0)[lb]{\smash{{\SetFigFont{9}{10.8}{\rmdefault}{\mddefault}{\updefault}$v_{j+1}$}}}}
\put(3556,-1141){\makebox(0,0)[lb]{\smash{{\SetFigFont{9}{10.8}{\rmdefault}{\mddefault}{\updefault}$v_{j}$}}}}
\put(2836,-1141){\makebox(0,0)[lb]{\smash{{\SetFigFont{9}{10.8}{\rmdefault}{\mddefault}{\updefault}$b_{v_{i-1}}$}}}}
\put(1486,-1141){\makebox(0,0)[lb]{\smash{{\SetFigFont{9}{10.8}{\rmdefault}{\mddefault}{\updefault}$b_{v_{j+1}}$}}}}
\put(856,-1141){\makebox(0,0)[lb]{\smash{{\SetFigFont{9}{10.8}{\rmdefault}{\mddefault}{\updefault}$v_{i}$}}}}
\put(496,-1141){\makebox(0,0)[lb]{\smash{{\SetFigFont{9}{10.8}{\rmdefault}{\mddefault}{\updefault}$v_{i-1}$}}}}
\put(5851,-1141){\makebox(0,0)[lb]{\smash{{\SetFigFont{9}{10.8}{\rmdefault}{\mddefault}{\updefault}$v_{i}$}}}}
\put(5491,-1141){\makebox(0,0)[lb]{\smash{{\SetFigFont{9}{10.8}{\rmdefault}{\mddefault}{\updefault}$v_{i-1}$}}}}
\put(7066,-1591){\makebox(0,0)[lb]{\smash{{\SetFigFont{11}{13.2}{\rmdefault}{\mddefault}{\updefault}$(b)$}}}}
\put(2161,-1591){\makebox(0,0)[lb]{\smash{{\SetFigFont{11}{13.2}{\rmdefault}{\mddefault}{\updefault}$(a)$}}}}
\put(6661,-556){\makebox(0,0)[lb]{\smash{{\SetFigFont{8}{9.6}{\rmdefault}{\mddefault}{\updefault}$X_2$}}}}
\put(7471,-556){\makebox(0,0)[lb]{\smash{{\SetFigFont{8}{9.6}{\rmdefault}{\mddefault}{\updefault}$X_1$}}}}
\end{picture}%

%% file: notx1x2.pdf_t
\begin{picture}(0,0)%
\includegraphics{notx1x2.pdf}%
\end{picture}%
\setlength{\unitlength}{3108sp}%
\begingroup\makeatletter\ifx\SetFigFont\undefined%
\gdef\SetFigFont#1#2#3#4#5{%
  \reset@font\fontsize{#1}{#2pt}%
  \fontfamily{#3}\fontseries{#4}\fontshape{#5}%
  \selectfont}%
\fi\endgroup%
\begin{picture}(9564,2389)(79,-1313)
\put(2656,-871){\makebox(0,0)[lb]{\smash{{\SetFigFont{9}{10.8}{\rmdefault}{\mddefault}{\updefault}$b_{v_{j+1}}$}}}}
\put(2386,-871){\makebox(0,0)[lb]{\smash{{\SetFigFont{9}{10.8}{\rmdefault}{\mddefault}{\updefault}$b_{v_j}$}}}}
\put(3781,-871){\makebox(0,0)[lb]{\smash{{\SetFigFont{9}{10.8}{\rmdefault}{\mddefault}{\updefault}$v_{j+1}$}}}}
\put(3556,-871){\makebox(0,0)[lb]{\smash{{\SetFigFont{9}{10.8}{\rmdefault}{\mddefault}{\updefault}$v_j$}}}}
\put(496,-871){\makebox(0,0)[lb]{\smash{{\SetFigFont{9}{10.8}{\rmdefault}{\mddefault}{\updefault}$v_{i-1}$}}}}
\put(856,-871){\makebox(0,0)[lb]{\smash{{\SetFigFont{9}{10.8}{\rmdefault}{\mddefault}{\updefault}$v_i$}}}}
\put(1576,-871){\makebox(0,0)[lb]{\smash{{\SetFigFont{9}{10.8}{\rmdefault}{\mddefault}{\updefault}$b_{v_{i-1}}$}}}}
\put(2206,-871){\makebox(0,0)[lb]{\smash{{\SetFigFont{9}{10.8}{\rmdefault}{\mddefault}{\updefault}$w$}}}}
\put(1981,-871){\makebox(0,0)[lb]{\smash{{\SetFigFont{9}{10.8}{\rmdefault}{\mddefault}{\updefault}$b_{v_i}$}}}}
\put(7816,-856){\makebox(0,0)[lb]{\smash{{\SetFigFont{9}{10.8}{\rmdefault}{\mddefault}{\updefault}$b_{v_{j+1}}$}}}}
\put(7546,-856){\makebox(0,0)[lb]{\smash{{\SetFigFont{9}{10.8}{\rmdefault}{\mddefault}{\updefault}$b_{v_i}$}}}}
\put(8941,-856){\makebox(0,0)[lb]{\smash{{\SetFigFont{9}{10.8}{\rmdefault}{\mddefault}{\updefault}$v_{j+1}$}}}}
\put(8716,-856){\makebox(0,0)[lb]{\smash{{\SetFigFont{9}{10.8}{\rmdefault}{\mddefault}{\updefault}$v_j$}}}}
\put(5656,-856){\makebox(0,0)[lb]{\smash{{\SetFigFont{9}{10.8}{\rmdefault}{\mddefault}{\updefault}$v_{i-1}$}}}}
\put(6016,-856){\makebox(0,0)[lb]{\smash{{\SetFigFont{9}{10.8}{\rmdefault}{\mddefault}{\updefault}$v_i$}}}}
\put(6736,-856){\makebox(0,0)[lb]{\smash{{\SetFigFont{9}{10.8}{\rmdefault}{\mddefault}{\updefault}$b_{v_{i-1}}$}}}}
\put(7366,-856){\makebox(0,0)[lb]{\smash{{\SetFigFont{9}{10.8}{\rmdefault}{\mddefault}{\updefault}$w$}}}}
\put(7141,-856){\makebox(0,0)[lb]{\smash{{\SetFigFont{9}{10.8}{\rmdefault}{\mddefault}{\updefault}$b_{v_j}$}}}}
\put(7336,929){\makebox(0,0)[lb]{\smash{{\SetFigFont{9}{10.8}{\rmdefault}{\mddefault}{\updefault}$m$}}}}
\put(7246,-1231){\makebox(0,0)[lb]{\smash{{\SetFigFont{11}{13.2}{\rmdefault}{\mddefault}{\updefault}$(b)$}}}}
\put(2026,-1231){\makebox(0,0)[lb]{\smash{{\SetFigFont{11}{13.2}{\rmdefault}{\mddefault}{\updefault}$(a)$}}}}
\put(2161,929){\makebox(0,0)[lb]{\smash{{\SetFigFont{9}{10.8}{\rmdefault}{\mddefault}{\updefault}$m$}}}}
\end{picture}%

%% file: 1branchrule.pdf_t
\begin{picture}(0,0)%
\includegraphics{1branchrule.pdf}%
\end{picture}%
\setlength{\unitlength}{4144sp}%
\begingroup\makeatletter\ifx\SetFigFont\undefined%
\gdef\SetFigFont#1#2#3#4#5{%
  \reset@font\fontsize{#1}{#2pt}%
  \fontfamily{#3}\fontseries{#4}\fontshape{#5}%
  \selectfont}%
\fi\endgroup%
\begin{picture}(3627,1390)(661,-674)
\put(1036,-151){\makebox(0,0)[lb]{\smash{{\SetFigFont{12}{14.4}{\rmdefault}{\mddefault}{\updefault}$B'$}}}}
\put(3286,569){\makebox(0,0)[lb]{\smash{{\SetFigFont{12}{14.4}{\rmdefault}{\mddefault}{\updefault}$\in G$}}}}
\put(676,-286){\makebox(0,0)[lb]{\smash{{\SetFigFont{12}{14.4}{\rmdefault}{\mddefault}{\updefault}$\underbrace{\ \ \ \ \ \ \ \ \ \ \ \ \ \ \ \ \ \ \ \ \ \ \ \ \ \ \  \ \ \ \ \ \ }$}}}}
\put(2431,-151){\makebox(0,0)[lb]{\smash{{\SetFigFont{12}{14.4}{\rmdefault}{\mddefault}{\updefault}$N_D$}}}}
\put(2161,-151){\makebox(0,0)[lb]{\smash{{\SetFigFont{12}{14.4}{\rmdefault}{\mddefault}{\updefault}$b$}}}}
\put(2071,389){\makebox(0,0)[lb]{\smash{{\SetFigFont{12}{14.4}{\rmdefault}{\mddefault}{\updefault}$N_G(b)$}}}}
\put(3421,-151){\makebox(0,0)[lb]{\smash{{\SetFigFont{12}{14.4}{\rmdefault}{\mddefault}{\updefault}$C$}}}}
\put(1576,-601){\makebox(0,0)[lb]{\smash{{\SetFigFont{12}{14.4}{\rmdefault}{\mddefault}{\updefault}$B_b$}}}}
\put(3421,-601){\makebox(0,0)[lb]{\smash{{\SetFigFont{12}{14.4}{\rmdefault}{\mddefault}{\updefault}$\sigma_H[C]$}}}}
\put(2746,-286){\makebox(0,0)[lb]{\smash{{\SetFigFont{12}{14.4}{\rmdefault}{\mddefault}{\updefault}$\underbrace{\ \ \ \ \ \ \ \ \ \ \ \ \ \ \ \ \ \ \ \  \ \ \ \ \ \ }$}}}}
\end{picture}%

%% file: 1branchrule1.pdf_t
\begin{picture}(0,0)%
\includegraphics{1branchrule1.pdf}%
\end{picture}%
\setlength{\unitlength}{4144sp}%
\begingroup\makeatletter\ifx\SetFigFont\undefined%
\gdef\SetFigFont#1#2#3#4#5{%
  \reset@font\fontsize{#1}{#2pt}%
  \fontfamily{#3}\fontseries{#4}\fontshape{#5}%
  \selectfont}%
\fi\endgroup%
\begin{picture}(7404,1398)(259,-494)
\put(1261,-421){\makebox(0,0)[lb]{\smash{{\SetFigFont{12}{14.4}{\rmdefault}{\mddefault}{\updefault}(a)}}}}
\put(6211,-376){\makebox(0,0)[lb]{\smash{{\SetFigFont{12}{14.4}{\rmdefault}{\mddefault}{\updefault}(c)}}}}
\put(3601,-421){\makebox(0,0)[lb]{\smash{{\SetFigFont{12}{14.4}{\rmdefault}{\mddefault}{\updefault}(b)}}}}
\put(4276,-61){\makebox(0,0)[lb]{\smash{{\SetFigFont{9}{10.8}{\rmdefault}{\mddefault}{\updefault}$w\in C$}}}}
\put(3781,-61){\makebox(0,0)[lb]{\smash{{\SetFigFont{9}{10.8}{\rmdefault}{\mddefault}{\updefault}$v\in N_D$}}}}
\put(2701,-61){\makebox(0,0)[lb]{\smash{{\SetFigFont{9}{10.8}{\rmdefault}{\mddefault}{\updefault}$u'\in B$}}}}
\put(3286,-61){\makebox(0,0)[lb]{\smash{{\SetFigFont{9}{10.8}{\rmdefault}{\mddefault}{\updefault}$u\in N_D$}}}}
\put(5041,-61){\makebox(0,0)[lb]{\smash{{\SetFigFont{9}{10.8}{\rmdefault}{\mddefault}{\updefault}$u'\in B$}}}}
\put(5626,-61){\makebox(0,0)[lb]{\smash{{\SetFigFont{9}{10.8}{\rmdefault}{\mddefault}{\updefault}$u\in N_D$}}}}
\put(6166,-61){\makebox(0,0)[lb]{\smash{{\SetFigFont{9}{10.8}{\rmdefault}{\mddefault}{\updefault}$v\in N_D$}}}}
\put(6706,-61){\makebox(0,0)[lb]{\smash{{\SetFigFont{9}{10.8}{\rmdefault}{\mddefault}{\updefault}$w\in C$}}}}
\put(7156,-61){\makebox(0,0)[lb]{\smash{{\SetFigFont{9}{10.8}{\rmdefault}{\mddefault}{\updefault}$w'\in C$}}}}
\put(1936,-61){\makebox(0,0)[lb]{\smash{{\SetFigFont{9}{10.8}{\rmdefault}{\mddefault}{\updefault}$w'\in C$}}}}
\put(361,-61){\makebox(0,0)[lb]{\smash{{\SetFigFont{9}{10.8}{\rmdefault}{\mddefault}{\updefault}$u\in N_D$}}}}
\put(946,-61){\makebox(0,0)[lb]{\smash{{\SetFigFont{9}{10.8}{\rmdefault}{\mddefault}{\updefault}$v\in N_D$}}}}
\put(1441,-61){\makebox(0,0)[lb]{\smash{{\SetFigFont{9}{10.8}{\rmdefault}{\mddefault}{\updefault}$w\in C$}}}}
\end{picture}%

%% file: bijoin_dec.pdf_t
\begin{picture}(0,0)%
\includegraphics{bijoin_dec.pdf}%
\end{picture}%
\setlength{\unitlength}{3158sp}%
\begingroup\makeatletter\ifx\SetFigFont\undefined%
\gdef\SetFigFont#1#2#3#4#5{%
  \reset@font\fontsize{#1}{#2pt}%
  \fontfamily{#3}\fontseries{#4}\fontshape{#5}%
  \selectfont}%
\fi\endgroup%
\begin{picture}(8225,1284)(839,-1189)
\put(1373,-136){\makebox(0,0)[lb]{\smash{{\SetFigFont{14}{16.8}{\rmdefault}{\mddefault}{\updefault}$B_1$}}}}
\put(7845,-136){\makebox(0,0)[lb]{\smash{{\SetFigFont{14}{16.8}{\rmdefault}{\mddefault}{\updefault}$B_2$}}}}
\put(1201,-736){\makebox(0,0)[lb]{\smash{{\SetFigFont{10}{12.0}{\rmdefault}{\mddefault}{\updefault}$B'_1$}}}}
\put(2626,-736){\makebox(0,0)[lb]{\smash{{\SetFigFont{10}{12.0}{\rmdefault}{\mddefault}{\updefault}$B'_2$}}}}
\put(4051,-736){\makebox(0,0)[lb]{\smash{{\SetFigFont{10}{12.0}{\rmdefault}{\mddefault}{\updefault}$B'_3$}}}}
\put(6526,-736){\makebox(0,0)[lb]{\smash{{\SetFigFont{10}{12.0}{\rmdefault}{\mddefault}{\updefault}$B'_{p-1}$}}}}
\put(7951,-736){\makebox(0,0)[lb]{\smash{{\SetFigFont{10}{12.0}{\rmdefault}{\mddefault}{\updefault}$B'_p$}}}}
\put(976,-1111){\makebox(0,0)[lb]{\smash{{\SetFigFont{10}{12.0}{\rmdefault}{\mddefault}{\updefault}$b_1$}}}}
\put(3826,-1111){\makebox(0,0)[lb]{\smash{{\SetFigFont{10}{12.0}{\rmdefault}{\mddefault}{\updefault}$b_{l_2+1}$}}}}
\put(6301,-1111){\makebox(0,0)[lb]{\smash{{\SetFigFont{10}{12.0}{\rmdefault}{\mddefault}{\updefault}$b_{l_{p-2}+1}$}}}}
\put(7126,-1111){\makebox(0,0)[lb]{\smash{{\SetFigFont{10}{12.0}{\rmdefault}{\mddefault}{\updefault}$b_{l_{p-1}}$}}}}
\put(7726,-1111){\makebox(0,0)[lb]{\smash{{\SetFigFont{10}{12.0}{\rmdefault}{\mddefault}{\updefault}$b_{l_{p-1}+1}$}}}}
\put(8626,-1111){\makebox(0,0)[lb]{\smash{{\SetFigFont{10}{12.0}{\rmdefault}{\mddefault}{\updefault}$b_{l_p}$}}}}
\put(2026,-1111){\makebox(0,0)[lb]{\smash{{\SetFigFont{10}{12.0}{\rmdefault}{\mddefault}{\updefault}$b_{l_1}$}}}}
\put(3376,-1111){\makebox(0,0)[lb]{\smash{{\SetFigFont{10}{12.0}{\rmdefault}{\mddefault}{\updefault}$b_{l_2}$}}}}
\put(4801,-1111){\makebox(0,0)[lb]{\smash{{\SetFigFont{10}{12.0}{\rmdefault}{\mddefault}{\updefault}$b_{l_3}$}}}}
\put(2401,-1111){\makebox(0,0)[lb]{\smash{{\SetFigFont{10}{12.0}{\rmdefault}{\mddefault}{\updefault}$b_{l_1+1}$}}}}
\end{picture}%

%% file: size.pdf_t
\begin{picture}(0,0)%
\includegraphics{size.pdf}%
\end{picture}%
\setlength{\unitlength}{3158sp}%
\begingroup\makeatletter\ifx\SetFigFont\undefined%
\gdef\SetFigFont#1#2#3#4#5{%
  \reset@font\fontsize{#1}{#2pt}%
  \fontfamily{#3}\fontseries{#4}\fontshape{#5}%
  \selectfont}%
\fi\endgroup%
\begin{picture}(6999,1065)(1489,-814)
\put(2626,-736){\makebox(0,0)[lb]{\smash{{\SetFigFont{10}{12.0}{\rmdefault}{\mddefault}{\updefault}$a_1$}}}}
\put(3826,-736){\makebox(0,0)[lb]{\smash{{\SetFigFont{10}{12.0}{\rmdefault}{\mddefault}{\updefault}$a_i$}}}}
\put(5926,-736){\makebox(0,0)[lb]{\smash{{\SetFigFont{10}{12.0}{\rmdefault}{\mddefault}{\updefault}$a_{i+1}$}}}}
\put(7126,-736){\makebox(0,0)[lb]{\smash{{\SetFigFont{10}{12.0}{\rmdefault}{\mddefault}{\updefault}$a_p$}}}}
\put(7876,-61){\makebox(0,0)[lb]{\smash{{\SetFigFont{10}{12.0}{\rmdefault}{\mddefault}{\updefault}$R_{p+1}$}}}}
\put(1951,-61){\makebox(0,0)[lb]{\smash{{\SetFigFont{10}{12.0}{\rmdefault}{\mddefault}{\updefault}$L_0$}}}}
\put(4876,-61){\makebox(0,0)[lb]{\smash{{\SetFigFont{10}{12.0}{\rmdefault}{\mddefault}{\updefault}$S_i$}}}}
\end{picture}%

%% file: forbiddenbcg.pdf_t
\begin{picture}(0,0)%
\includegraphics{forbiddenbcg.pdf}%
\end{picture}%
\setlength{\unitlength}{2901sp}%
\begingroup\makeatletter\ifx\SetFigFont\undefined%
\gdef\SetFigFont#1#2#3#4#5{%
  \reset@font\fontsize{#1}{#2pt}%
  \fontfamily{#3}\fontseries{#4}\fontshape{#5}%
  \selectfont}%
\fi\endgroup%
\begin{picture}(6894,1499)(1279,-1489)
\put(1396,-1411){\makebox(0,0)[lb]{\smash{{\SetFigFont{8}{9.6}{\rmdefault}{\mddefault}{\updefault}$2K_2$}}}}
\put(2791,-1411){\makebox(0,0)[lb]{\smash{{\SetFigFont{8}{9.6}{\rmdefault}{\mddefault}{\updefault}$K_3$}}}}
\put(4411,-1411){\makebox(0,0)[lb]{\smash{{\SetFigFont{8}{9.6}{\rmdefault}{\mddefault}{\updefault}$C_5$}}}}
\put(5941,-1411){\makebox(0,0)[lb]{\smash{{\SetFigFont{8}{9.6}{\rmdefault}{\mddefault}{\updefault}$C_4$}}}}
\put(7561,-1411){\makebox(0,0)[lb]{\smash{{\SetFigFont{8}{9.6}{\rmdefault}{\mddefault}{\updefault}$3K_1$}}}}
\end{picture}%

%% file: Kjoindec_biclique.pdf_t
\begin{picture}(0,0)%
\includegraphics{Kjoindec_biclique.pdf}%
\end{picture}%
\setlength{\unitlength}{3315sp}%
\begingroup\makeatletter\ifx\SetFigFont\undefined%
\gdef\SetFigFont#1#2#3#4#5{%
  \reset@font\fontsize{#1}{#2pt}%
  \fontfamily{#3}\fontseries{#4}\fontshape{#5}%
  \selectfont}%
\fi\endgroup%
\begin{picture}(5328,875)(644,-319)
\put(1621, 98){\makebox(0,0)[lb]{\smash{{\SetFigFont{10}{12.0}{\rmdefault}{\mddefault}{\updefault}$M$}}}}
\put(2186, 97){\makebox(0,0)[lb]{\smash{{\SetFigFont{10}{12.0}{\rmdefault}{\mddefault}{\updefault}$B_R$}}}}
\put(1126,397){\makebox(0,0)[lb]{\smash{{\SetFigFont{10}{12.0}{\rmdefault}{\mddefault}{\updefault}$B$}}}}
\put(2988, 90){\makebox(0,0)[lb]{\smash{{\SetFigFont{10}{12.0}{\rmdefault}{\mddefault}{\updefault}$N$}}}}
\put(5150, 97){\makebox(0,0)[lb]{\smash{{\SetFigFont{10}{12.0}{\rmdefault}{\mddefault}{\updefault}$C$}}}}
\put(3806, 97){\makebox(0,0)[lb]{\smash{{\SetFigFont{10}{12.0}{\rmdefault}{\mddefault}{\updefault}$R$}}}}
\put(933,100){\makebox(0,0)[lb]{\smash{{\SetFigFont{10}{12.0}{\rmdefault}{\mddefault}{\updefault}$B_L$}}}}
\end{picture}%

%% file: sizebcc.pdf_t
\begin{picture}(0,0)%
\includegraphics{sizebcc.pdf}%
\end{picture}%
\setlength{\unitlength}{3552sp}%
\begingroup\makeatletter\ifx\SetFigFont\undefined%
\gdef\SetFigFont#1#2#3#4#5{%
  \reset@font\fontsize{#1}{#2pt}%
  \fontfamily{#3}\fontseries{#4}\fontshape{#5}%
  \selectfont}%
\fi\endgroup%
\begin{picture}(5355,1480)(211,-1114)
\put(226,-436){\makebox(0,0)[lb]{\smash{{\SetFigFont{11}{13.2}{\rmdefault}{\mddefault}{\updefault}$H_1$}}}}
\put(5551,-436){\makebox(0,0)[lb]{\smash{{\SetFigFont{11}{13.2}{\rmdefault}{\mddefault}{\updefault}$H_2$}}}}
\put(1576,-1036){\makebox(0,0)[lb]{\smash{{\SetFigFont{11}{13.2}{\rmdefault}{\mddefault}{\updefault}$A'_1$}}}}
\put(3976,-1036){\makebox(0,0)[lb]{\smash{{\SetFigFont{11}{13.2}{\rmdefault}{\mddefault}{\updefault}$A'_2$}}}}
\end{picture}%

%% file: papier.bbl
\begin{thebibliography}{10}

\bibitem{BPP10}
S.~Bessy, C.~Paul, and A.~Perez.
\newblock Polynomial kernels for 3-leaf power graph modification problems.
\newblock {\em Discrete Applied Mathematics}, 158(16):1732--1744, 2010.

\bibitem{Bod09}
H.~L. Bodlaender.
\newblock Kernelization: New upper and lower bound techniques.
\newblock In {\em IWPEC}, pages 17--37, 2009.

\bibitem{BDFH09}
H.~L. Bodlaender, R.~G. Downey, M.~R. Fellows, and D.~Hermelin.
\newblock On problems without polynomial kernels.
\newblock {\em J. Comput. Syst. Sci}, 75(8):423--434, 2009.

\bibitem{Cai96}
L.~Cai.
\newblock Fixed-parameter tractability of graph modification problems for
  hereditary properties.
\newblock {\em Inf. Process. Lett}, 58(4):171--176, 1996.

\bibitem{CM10}
J.~Chen and J.~Meng.
\newblock A $2k$ kernel for the cluster editing problem.
\newblock In {\em COCOON}, volume 6196 of {\em LNCS}, pages 459--468, 2010.

\bibitem{Corn04}
Derek~G. Corneil.
\newblock A simple 3-sweep lbfs algorithm for the recognition of unit interval
  graphs.
\newblock {\em Discrete Appl. Math.}, 138:371--379, April 2004.

\bibitem{Corn95}
Derek~G. Corneil, Hiryoung Kim, Sridhar Natarajan, Stephan Olariu, and Alan~P.
  Sprague.
\newblock Simple linear time recognition of unit interval graphs.
\newblock {\em Information Processing Letters}, 55(2):99 -- 104, 1995.

\bibitem{DFRS04}
F.~K. H.~A. Dehne, M.~R. Fellows, F.~A. Rosamond, and P.~Shaw.
\newblock Greedy localization, iterative compression, modeled crown reductions:
  New {FPT} techniques, an improved algorithm for set splitting, and a novel 2k
  kernelization for vertex cover.
\newblock In {\em IWPEC}, volume 3162 of {\em LNCS}, pages 271--280, 2004.

\bibitem{DF99}
R.G. Downey and M.R. Fellows.
\newblock {\em Parameterized complexity}.
\newblock Springer, 1999.

\bibitem{GKS94}
M.~C. Golumbic, H.~Kaplan, and R.~Shamir.
\newblock On the complexity of {DNA} physical mapping.
\newblock {\em ADVAM: Advances in Applied Mathematics}, 15, 1994.

\bibitem{GPP10}
S.~Guillemot, C.~Paul, and A.~Perez.
\newblock On the (non-)existence of polynomial kernels for $p_l$-free edge
  modification problems.
\newblock In {\em IPEC}, volume 6478 of {\em LNCS}, pages 147--157, 2010.

\bibitem{Guo07}
J.~Guo.
\newblock Problem kernels for {NP}-complete edge deletion problems: Split and
  related graphs.
\newblock In {\em ISAAC}, volume 4835 of {\em LNCS}, pages 915--926, 2007.

\bibitem{HSS01}
P.~Hell, R.~Shamir, and R.~Sharan.
\newblock A fully dynamic algorithm for recognizing and representing proper
  interval graphs.
\newblock {\em SIAM Journal on Computing}, 31(1):289--305, 2001.

\bibitem{KST94}
H.~Kaplan, R.~Shamir, and R.~E. Tarjan.
\newblock Tractability of parameterized completion problems on chordal and
  interval graphs: Minimum fill-in and physical mapping.
\newblock In {\em FOCS}, pages 780--791, 1994.

\bibitem{KST99}
H.~Kaplan, R.~Shamir, and R.~E. Tarjan.
\newblock Tractability of parameterized completion problems on chordal,
  strongly chordal, and proper interval graphs.
\newblock {\em SIAM J. Comput}, 28(5):1906--1922, 1999.

\bibitem{KW09}
S.~Kratsch and M.~Wahlstr{\"o}m.
\newblock Two edge modification problems without polynomial kernels.
\newblock In {\em IWPEC}, volume 5917 of {\em LNCS}, pages 264--275. Springer,
  2009.

\bibitem{LO93}
P.~J. Looges and S.~Olariu.
\newblock Optimal greedy algorithms for indifference graphs.
\newblock {\em Computers \& Mathematics with Applications}, 25(7):15 -- 25,
  1993.

\bibitem{Man08}
F.~Mancini.
\newblock {\em Graph modification problems related to graph classes}.
\newblock PhD thesis, University of Bergen, Norway, 2008.

\bibitem{Nie06}
R.~Niedermeier.
\newblock {\em Invitation to fixed parameter algorithms}, volume~31 of {\em
  Oxford Lectures Series in Mathematics and its Applications}.
\newblock Oxford University Press, 2006.

\bibitem{NR00}
R.~Niedermeier and P.~Rossmanith.
\newblock A general method to speed up fixed-parameter-tractable algorithms.
\newblock {\em Inf. Process. Lett}, 73(3-4):125--129, 2000.

\bibitem{SST04}
R.~Shamir, R.~Sharan, and D.~Tsur.
\newblock Cluster graph modification problems.
\newblock {\em Discrete Applied Mathematics}, 144(1-2):173--182, 2004.

\bibitem{Sha02}
R.~Sharan.
\newblock {\em Graph modification problems and their applications to genomic
  research}.
\newblock PhD thesis, Tel-Aviv University, 2002.

\bibitem{TY84}
R.~E. Tarjan and M.~Yannakakis.
\newblock Simple linear-time algorithms to test chordality of graphs, test
  acyclicity of hypergraphs, and selectively reduce acyclic hypergraphs.
\newblock {\em SIAM J. Comput}, 13(3):566--579, 1984.

\bibitem{Tho10}
S.~Thomass{\'e}.
\newblock A $4k^2$ kernel for feedback vertex set.
\newblock {\em ACM Transactions on Algorithms}, 6(2), 2010.

\bibitem{Vil10a}
Y.~Villanger.
\newblock \url{www.lirmm.fr/~paul/ANR/CIRM-TALKS-2010/Villanger-cirm-2010.pdf},
  2010.

\bibitem{Vil10}
Y.~Villanger.
\newblock Proper interval vertex deletion.
\newblock In {\em IPEC}, volume 6478 of {\em LNCS}, pages 228--238, 2010.

\bibitem{VHPT09}
Y.~Villanger, P.~Heggernes, C.~Paul, and J.~A. Telle.
\newblock Interval completion is fixed parameter tractable.
\newblock {\em SIAM J. Comput}, 38(5):2007--2020, 2009.

\bibitem{Weg67}
G.~Wegner.
\newblock {\em Eigenschaften der nerven homologische-einfactor familien in
  $R^n$}.
\newblock PhD thesis, Universit\"at Gottigen, Gottingen, Germany, 1967.

\bibitem{Yan81}
M.~Yannakakis.
\newblock Computing the minimum fill-in is {NP-Complete}.
\newblock {\em SIAM J. Alg. and Discr. Meth.}, 2(1):77--79, 1981.

\end{thebibliography}
